\documentclass[a4paper,UKenglish]{lipics-v2021}
\pdfoutput=1

\usepackage{amsmath} 
\usepackage{amssymb}
\usepackage{tikz}
\usepackage[noend]{algpseudocode}
\usepackage{algorithm}
\usetikzlibrary{snakes}
\usetikzlibrary{shapes.geometric}
\usetikzlibrary{decorations.pathreplacing}

\algnewcommand\algorithmicforeach{\textbf{for each}}
\algdef{S}[FOR]{ForEach}[1]{\algorithmicforeach\ #1\ \algorithmicdo}

\usepackage{thmtools}

\usepackage{hyperref}

\usepackage{cleveref}

\def\ca#1{{\mathcal{#1}}}

\usepackage[appendix=inline]{apxproof}
\ifx\proof\inlineproof
  \def\apxmark{}
\else
  \def\apxmark{{\bf$\!$*}}
\fi
\newtheoremrep{propositionx2}[theorem]{Proposition}
\newtheoremrep{proposition2}[theorem]{Proposition}
\newtheoremrep{theorem2}[theorem]{Theorem}
\newtheoremrep{lemma2}[theorem]{Lemma}
\newtheoremrep{claim2}[theorem]{Claim}

\title{Automorphisms of Set Families and of Families of Cliques in an Interval Graph in FPT Time}
\titlerunning{Automorphisms of Clique Families in an Interval Graph in FPT}

\author{Deniz A\u{g}ao\u{g}lu \c{C}a\u{g}{\i}r{\i}c{\i}}{Masaryk University, Brno, Czech Republic}{agaoglu@mail.muni.cz}{https://orcid.org/0000-0002-1691-0434}{}
\author{Petr Hlin\v en\'y}{Masaryk University, Brno, Czech Republic}{hlineny@fi.muni.cz}{https://orcid.org/0000-0003-2125-1514}{}
\authorrunning{D.~A\u{g}ao\u{g}lu and P.~Hlin\v en\'y}
\Copyright{Deniz A\u{g}ao\u{g}lu and Petr Hlin\v en\'y}



\category{}
\relatedversion{}
\supplement{}

\funding{Supported by the {Czech Science Foundation}, project no.~{20-04567S}.}

\EventLogo{}
\SeriesVolume{1}
\ArticleNo{XX}
\ifx\proof\inlineproof
\nolinenumbers 
\hideLIPIcs  
\else
\fi

\begin{document}
	
	\maketitle
	
	\begin{abstract}
		We consider the following problem closely related to graph isomorphism.
		In a simplified version, the task is to compute the automorphism group of a given set family (or a hypergraph),
		that is, the group of all automorphisms of the given sets which are compatible with some permutation of their elements.
		In a general setting, the set family in question is a collection of cliques (called marked cliques) of a given interval graph,
		and the task is to compute the group of all permutations of the cliques which result from some automorpism of the underlying interval graph.
		This problem is obviously at least as hard as the graph isomorphism (GI-hard) already in the simplified version --
		consider the set family of edges of a graph, and we give an FPT-time algorithm parameterized by the
		maximum number of sets in the family which are incomparable by inclusion (its antichain size).

		To our best knowledge, the general version of the problem has not been formulated in the literature so far.
		The problem has been inspired by the research of special cases of the isomorphism problem of chordal graphs; namely,
		the simplified set-family version is the core of our FPT algorithm for the isomorphism of so-called $S_d$-graphs [MFCS 2021],
		and the general version extends and improves a cumbersome technical step in our FPT algorithm for the isomorphism of chordal graphs of bounded leafage [WALCOM 2022].
		The new algorithm combines two classical tools -- PQ-trees of interval graphs and Babai's tower-of-groups, in a nontrivial way.
	\end{abstract}

\section{Introduction}\label{introduction}
	
	The \emph{graph isomorphism problem} is to determine whether the two given graphs are isomorphic, denoted by $G\simeq H$;
	i.e., to decide whether there exists a bijection between the vertex sets, $f:V(G)\to V(H)$, such that $f$ preserves the edges,
	$\{u,v\}\in E(G)$ $\iff$ $\{f(u),f(v)\}\in E(H)$ for all~$u,v\in V(G)$.
	All isomorphisms of a graph $G$ onto itself form a permutation group called the {\em automorphism group of~$G$}.
	
	Graph isomorphism is in a sense a quite special problem in computer science;
	on one hand, under some widely-believed complexity-theoretic assumptions, it can be shown that
	graph isomorphism is not an NP-hard problem, while on the other hand,
	a polynomial-time algorithm for graph isomorphism is still elusive (and not
	everybody expects existence of such algorithm).
	It has actually defined its own complexity class \emph{GI} of the problems
	which are reducible in polynomial time to graph isomorphism.
	The current state of the art is a quasi-polynomial algorithm of Babai~\cite{DBLP:conf/stoc/Babai16}.
	Nevertheless, the problem has been shown to be solvable efficiently for various
	natural graph classes such as trees, planar and interval graphs \cite{AHU,planarLinear,recogIntervalLinear}.

	For a set family $\ca X\subseteq 2^M$, the {\em automorphism group of~$\ca X$} is the group of all permutations $\sigma$ of $\ca X$
	for which there exists a permutation $\pi$ of $M$ such that $\sigma(X)=\pi(X)$ holds for all~$X\in\ca X$.
	If there are no restrictions on the considered permutations of the ground set~$M$, the condition on $\sigma$ can be simply translated
	as that $\sigma$ preserves the cardinalities of all intersections of sets from $\ca X$ (cf.~Lemma~\ref{lem:preciseVenn}).
	Note that this problem is at least as hard as computing the automorphism group of a graph (that is, {\em GI-hard\/});
	simply take the edge set of the graph as the considered family $\ca X$.
	In our approach, we focus on such instances in which the maximum number of inclusion-incomparable sets in $\ca X$, i.e., the {\em antichain size of~$\ca X$}, is a parameter.

	In an extended view of the problem, we have a permutation group $\Gamma$ over the ground set~$M$, and we ask to compute the group $\Delta$
	of all permutations $\sigma$ of $\ca X$ for which there exists $\pi\in\Gamma$ such that $\sigma(X)=\pi(X)$ holds for all~$X\in\ca X$.
	We call such $\Delta$ the {\em action of $\Gamma$ on $\ca X$}.
	As the example of $\ca X$ being the edge set of a graph shows, the computation of $\Delta$ is nontrivial even if we can efficiently compute with the group~$\Gamma$.
	We are here interested in the case that $\Gamma$ is the automorphism group of some graph (and then $\ca X$ are subsets of vertices of this graph, hereafter called marked sets)
	and, in particular, of an interval graph for which the automorphism group $\Gamma$ can be easily computed~\cite{recogIntervalLinear,AutMPQtrees}.

\subsubsection*{Problem definition and main result}

	A graph $G$ is an \emph{interval graph} if the vertex set of $G$ can be	mapped into some set of intervals on the real line such that two vertices of
	$G$ are adjacent if and only if the corresponding intervals intersect.
	
	\begin{definition}[Action of the marking-preserving automorphism group of an interval graph]\label{def:autommarked}~\rm\\
		The problem {\sc AutomMarkedINT}$(G;\,\ca A^1,\ldots,\ca A^m)$ is defined as follows:
		\begin{description}
			\item[Input] An interval graph $G$, and families $\ca A^1,\ldots,\ca A^m$ of nonempty subsets of $V(G)$ such that
			every set $A\in\ca A^i$ where $i\in\{1,\ldots,m\}$ induces a clique of~$G$.
			The sets $A\in\ca A^1\cup\ldots\cup\ca A^m$ are called the {\em marked sets} of $G$, and specially,
			$A\in\ca A^i$ is called a {\em marked set of color~$i$}.
			\item[Task] Compute the group $\Gamma$ of such permutations of $\ca A:=\ca A^1\cup\ldots\cup\ca A^m$ that are the
			{\em actions of the $(\ca A^1,\ldots,\ca A^m)$-preserving automorphisms} of~$G$.
			In other words, a permutation $\tau$ of $\ca A$ belongs to $\Gamma$, if and only if there exists an automorphism $\varrho$ of $G$
			such that, for every $i\in\{1,\ldots,m\}$ and all $A\in\ca A^i$, we have~$\tau(A)\in\ca A^i$ and~$\tau(A)=\varrho(A)$.
		\end{description}
		To stay on the more general side,
		we consider set families in the {\em multiset} setting, meaning that the same set $A$ may occur in a family $\ca A^i$ multiple times.
		In regard of the above stated condition, a permutation $\tau$ of $\ca A$ then obviously has to map $A\in\ca A$ to a set $\tau(A)=A'\in\ca A$
		such that the multiplicities of $A$ and of $A'$ in each of $\ca A^1,\ldots,\ca A^m$ are the same.
	\end{definition}

	Regarding computational complexity, we remark that by `computing a~group' we mean to output a set of its generators which is at most polynomially large by~\cite{furst}
	(while the permutation group itself is often exponentially large compared to the ground set). 
	It is also important to mention what is the input size of an {\sc AutomMarkedINT}$(G;\,\ca A^1,\ldots,\ca A^m)$ instance.
	If $G$ is an $n$-vertex graph, the number of sets in $\ca A$ may be up to exponential in~$n$.
	However, considering our parameter $a$ equal to the maximum antichain size of~$\ca A$, we easily get that there are at most $an$ distinct sets in $\ca A$ (at most $a$ of each cardinality between $1$ and~$n$).
	Hence we can always upper-bound the input size by $\big(|V(G)|+|E(G)|+\sum_{A\in\ca A}|A|\big) ~\in~ \ca O(a\cdot|V(G)|^2)$.
	
	\begin{theorem}\label{thm:autommarked}
		The problem {\sc AutomMarkedINT}$(G;\,\ca A^1,\ldots,\ca A^m)$ is solvable in FPT-time with respect to the parameter $a$ which is the maximum antichain size of $\ca A:=\ca A^1\cup\ldots\cup\ca A^m$.
	\end{theorem}
	
	In the course of proving \Cref{thm:autommarked}, we combine classical PQ-trees for capturing the internal structure of interval graphs \cite{recogIntervalLinear} (\Cref{sec:intPQ}),
	and a colored extension of the algorithm for computing the automorphism group of a set family of bounded antichain size from \cite{aaolu2019isomorphism} which is based
	on another classical tool -- Babai's tower-of-groups procedure~\cite{babai-bdcm} (\Cref{sec:autsets}).
	The full proof is finished in \Cref{sec:setstoint}.
	
	We add that our algorithm solving \Cref{thm:autommarked} does not first compute the $(\ca A^1,\ldots,\ca A^m)$-preserving automorphism group of~$G$, not even implicitly;
	with our approach it is easier this way.
	However, the algorithm can be straightforwardly extended to compute also the full $(\ca A^1,\ldots,\ca A^m)$-preserving automorphism group of~$G$ in FPT time with~respect~to~$a$.

\subsubsection*{Motivation of the problem}

	The problem given by Definition~\ref{def:autommarked} is not artificial, but naturally follows from our recent research of special cases of the isomorphism problem of chordal graphs.
	In detail, the simplified set-family version of it (with no colors and no underlying graph) is implicitly used and solved
	in the core of an FPT algorithm for the isomorphism of so-called $S_d$-graphs~\cite{aaolu2019isomorphism}.
	A complicated extension of the algorithm of~\cite{aaolu2019isomorphism} was used recently in~\cite{DBLP:conf/walcom/CagiriciH22}
	to solve the isomorphism problem of chordal graphs of bounded leafage also in FPT.
	Subsequently to that, we have formulated and solved the problem of Definition~\ref{def:autommarked},
	which is among other things a handy replacement and rigorous improvement over the cumbersome technical part of the algorithm~\cite{DBLP:conf/walcom/CagiriciH22}.

	To informally explain the mentioned use of the algorithm of \Cref{thm:autommarked}, we summarize that \cite{DBLP:conf/walcom/CagiriciH22}
	defines and proves a canonical (i.e., isomorphism-invariant) decomposition procedure of a chordal graph into a collection of interval graphs over clique-cutsets.
	The condition of bounded leafage of the graph implies that these clique-cutsets in every component of the decomposition have bounded-size antichains by inclusion.
	Since isomorphism of interval graphs (the decomposed components) is well-understood (\Cref{sec:intPQ}), it then remains 
	to efficiently handle the actions of automorphisms of the components on their clique-cutsets to finish the algorithm 
	for isomorphism of chordal graphs of bounded leafage (see \cite{DBLP:conf/walcom/CagiriciH22} for more details).

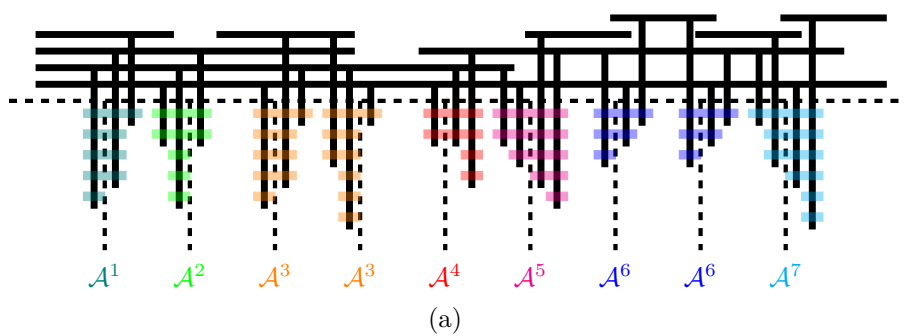
\begin{figure}[tb]
		\centering
		\begin{subfigure}[t]{1\linewidth}
			\centering
			\begin{tikzpicture}[xscale=1.4,yscale=1.1]
				
				\draw[black, dashed, ultra thick] (-4.25,0) -- (4.25,0);
				
				\draw[black, dashed, ultra thick] (-3.35,0) -- (-3.35,-1.8);
				\draw[black, dashed, ultra thick] (-2.55,0) -- (-2.55,-1.8);
				\draw[black, dashed, ultra thick] (-1.75,0) -- (-1.75,-1.8);
				\draw[black, dashed, ultra thick] (-0.95,0) -- (-0.95,-1.8);
				\draw[black, dashed, ultra thick] (-0.15,0) -- (-0.15,-1.8);
				\draw[black, dashed, ultra thick] (0.65,0) -- (0.65,-1.8);
				\draw[black, dashed, ultra thick] (1.45,0) -- (1.45,-1.8);
				\draw[black, dashed, ultra thick] (2.25,0) -- (2.25,-1.8);
				\draw[black, dashed, ultra thick] (3.05,0) -- (3.05,-1.8);
				
				\draw[black, line width=2.6] (-4,0.2) -- (4,0.2);
				\draw[black, line width=2.6] (-4,0.4) -- (0.5,0.4);
				\draw[black, line width=2.6] (-4,0.6) -- (-1,0.6);
				\draw[black, line width=2.6] (-0.4,0.6) -- (3.6,0.6);
				\draw[black, line width=2.6] (-4,0.8) -- (-2.7,0.8);
				\draw[black, line width=2.6] (-2.3,0.8) -- (-1,0.8);
				\draw[black, line width=2.6] (0.6,0.8) -- (1.6,0.8);
				\draw[black, line width=2.6] (2.2,0.8) -- (3.2,0.8);
				\draw[black, line width=2.6] (1.4,1) -- (2.4,1);
				\draw[black, line width=2.6] (3,1) -- (4,1);
				
				\draw[black, line width=2.6] (-3.45,0.4) -- (-3.45,-1.3);
				\draw[black, line width=2.6] (-3.25,0.6) -- (-3.25,-1.05);
				\draw[black, line width=2.6] (-3.1,0.8) -- (-3.1,-0.3);
				
				\draw [draw=teal, fill=teal, opacity=0.4] (-3.55,-0.1) rectangle (-3,-0.2);
				\draw [draw=teal, fill=teal, opacity=0.4] (-3.55,-0.35) rectangle (-3.15,-0.45);
				\draw [draw=teal, fill=teal, opacity=0.4] (-3.55,-0.6) rectangle (-3.15,-0.7);
				\draw [draw=teal, fill=teal, opacity=0.4] (-3.55,-0.85) rectangle (-3.15,-0.95);
				\draw [draw=teal, fill=teal, opacity=0.4] (-3.55,-1.1) rectangle (-3.35,-1.2);
				
				\node (1) at (-3.35,-2.1) {\textcolor{teal}{$\mathcal{A}^1$}};	
				
				\draw[black, line width=2.6] (-2.65,0.4) -- (-2.65,-1.3);
				\draw[black, line width=2.6] (-2.45,0.6) -- (-2.45,-0.55);
				\draw[black, line width=2.6] (-2.8,0.2) -- (-2.8,-0.55);
				
				\draw [draw=green, fill=green, opacity=0.4] (-2.9,-0.1) rectangle (-2.35,-0.2);
				\draw [draw=green, fill=green, opacity=0.4] (-2.9,-0.35) rectangle (-2.35,-0.45);
				\draw [draw=green, fill=green, opacity=0.4] (-2.55,-0.6) rectangle (-2.75,-0.7);
				\draw [draw=green, fill=green, opacity=0.4] (-2.55,-0.85) rectangle (-2.75,-0.95);
				\draw [draw=green, fill=green, opacity=0.4] (-2.55,-1.1) rectangle (-2.75,-1.2);	
				
				\node (2) at (-2.55,-2.1) {\textcolor{green}{$\mathcal{A}^2$}};
				
				\draw[black, line width=2.6] (-1.85,0.2) -- (-1.85,-1.3);
				\draw[black, line width=2.6] (-1.65,0.8) -- (-1.65,-1.05);
				\draw[black, line width=2.6] (-1.5,0.4) -- (-1.5,-0.3);
				
				\draw [draw=orange, fill=orange, opacity=0.4] (-1.95,-0.1) rectangle (-1.4,-0.2);
				\draw [draw=orange, fill=orange, opacity=0.4] (-1.95,-0.35) rectangle (-1.55,-0.45);
				\draw [draw=orange, fill=orange, opacity=0.4] (-1.95,-0.6) rectangle (-1.55,-0.7);
				\draw [draw=orange, fill=orange, opacity=0.4] (-1.95,-0.85) rectangle (-1.55,-0.95);
				\draw [draw=orange, fill=orange, opacity=0.4] (-1.95,-1.1) rectangle (-1.75,-1.2);
				
				\node (3) at (-1.75,-2.1) {\textcolor{orange}{$\mathcal{A}^3$}};
				
				\draw[black, line width=2.6] (-0.85,0.2) -- (-0.85,-0.3);
				\draw[black, line width=2.6] (-1.2,0.8) -- (-1.2,-0.8);
				\draw[black, line width=2.6] (-1.05,0.4) -- (-1.05,-1.55);
				
				\draw [draw=orange, fill=orange, opacity=0.4] (-0.75,-0.1) rectangle (-1.3,-0.2);
				\draw [draw=orange, fill=orange, opacity=0.4] (-1.3,-0.35) rectangle (-0.95,-0.45);
				\draw [draw=orange, fill=orange, opacity=0.4] (-1.3,-0.6) rectangle (-0.95,-0.7);
				\draw [draw=orange, fill=orange, opacity=0.4] (-1.15,-0.85) rectangle (-0.95,-0.95);
				\draw [draw=orange, fill=orange, opacity=0.4] (-1.15,-1.1) rectangle (-0.95,-1.2);
				\draw [draw=orange, fill=orange, opacity=0.4] (-1.15,-1.35) rectangle (-0.95,-1.45);
				
				\node (4) at (-0.95,-2.1) {\textcolor{orange}{$\mathcal{A}^3$}};
				
				\draw[black, line width=2.6] (-0.25,0.2) -- (-0.25,-0.55);
				\draw[black, line width=2.6] (-0.05,0.4) -- (-0.05,-0.55);
				\draw[black, line width=2.6] (0.1,0.6) -- (0.1,-1.05);
				
				\draw [draw=red, fill=red, opacity=0.4] (-0.35,-0.1) rectangle (0.2,-0.2);
				\draw [draw=red, fill=red, opacity=0.4] (-0.35,-0.35) rectangle (0.2,-0.45);
				\draw [draw=red, fill=red, opacity=0.4] (0,-0.6) rectangle (0.2,-0.7);
				\draw [draw=red, fill=red, opacity=0.4] (0,-0.85) rectangle (0.2,-0.95);
				
				\node (5) at (-0.15,-2.1) {\textcolor{red}{$\mathcal{A}^4$}};
				
				\draw[black, line width=2.6] (0.55,0.2) -- (0.55,-0.8);
				\draw[black, line width=2.6] (0.4,0.4) -- (0.4,-0.55);
				\draw[black, line width=2.6] (0.75,0.8) -- (0.75,-1.05);
				\draw[black, line width=2.6] (0.9,0.6) -- (0.9,-1.3);
				
				\draw [draw=magenta, fill=magenta, opacity=0.4] (0.3,-0.1) rectangle (1,-0.2);
				\draw [draw=magenta, fill=magenta, opacity=0.4] (0.3,-0.35) rectangle (1,-0.45);
				\draw [draw=magenta, fill=magenta, opacity=0.4] (0.45,-0.6) rectangle (1,-0.7);
				\draw [draw=magenta, fill=magenta, opacity=0.4] (0.65,-0.85) rectangle (1,-0.95);
				\draw [draw=magenta, fill=magenta, opacity=0.4] (0.8,-1.1) rectangle (1,-1.2);
				
				\node (6) at (0.65,-2.1) {\textcolor{magenta}{$\mathcal{A}^5$}};
				
				\draw[black, line width=2.6] (1.35,0.6) -- (1.35,-0.8);
				\draw[black, line width=2.6] (1.55,0.2) -- (1.55,-0.55);
				\draw[black, line width=2.6] (1.7,1) -- (1.7,-0.3);
				\draw [draw=blue, fill=blue, opacity=0.4] (1.25,-0.1) rectangle (1.8,-0.2);
				\draw [draw=blue, fill=blue, opacity=0.4] (1.25,-0.35) rectangle (1.65,-0.45);
				\draw [draw=blue, fill=blue, opacity=0.4] (1.25,-0.6) rectangle (1.45,-0.7);
				
				\node (7) at (1.45,-2.1) {\textcolor{blue}{$\mathcal{A}^6$}};
				
				\draw[black, line width=2.6] (2.5,0.6) -- (2.5,-0.3);
				\draw[black, line width=2.6] (2.35,0.2) -- (2.35,-0.55);
				\draw[black, line width=2.6] (2.15,1) -- (2.15,-0.8);
				\draw [draw=blue, fill=blue, opacity=0.4] (2.05,-0.1) rectangle (2.6,-0.2);
				\draw [draw=blue, fill=blue, opacity=0.4] (2.05,-0.35) rectangle (2.45,-0.45);
				\draw [draw=blue, fill=blue, opacity=0.4] (2.05,-0.6) rectangle (2.25,-0.7);
				
				\node (8) at (2.25,-2.1) {\textcolor{blue}{$\mathcal{A}^6$}};
				
				\draw[black, line width=2.6] (2.8,0.6) -- (2.8,-0.3);
				\draw[black, line width=2.6] (2.95,0.8) -- (2.95,-0.8);
				\draw[black, line width=2.6] (3.15,0.2) -- (3.15,-1.05);
				\draw[black, line width=2.6] (3.3,1) -- (3.3,-1.55);
				\draw [draw=cyan, fill=cyan, opacity=0.4] (2.7,-0.1) rectangle (3.4,-0.2);
				\draw [draw=cyan, fill=cyan, opacity=0.4] (2.85,-0.35) rectangle (3.4,-0.45);
				\draw [draw=cyan, fill=cyan, opacity=0.4] (2.85,-0.6) rectangle (3.4,-0.7);
				\draw [draw=cyan, fill=cyan, opacity=0.4] (3.05,-0.85) rectangle (3.4,-0.95);
				\draw [draw=cyan, fill=cyan, opacity=0.4] (3.2,-1.1) rectangle (3.4,-1.2);
				\draw [draw=cyan, fill=cyan, opacity=0.4] (3.2,-1.35) rectangle (3.4,-1.45);
				
				\node (9) at (3.05,-2.1) {\textcolor{cyan}{$\mathcal{A}^7$}};
				
				\node (11) at (-0.15,-2.65) {(a)};	
				\label{a}
			\end{tikzpicture}
			
			\label{fig:UDVGd}
		\end{subfigure}
		
		\caption{(a) An interval graph (horizontally) depicted with black intervals, and the families $\mathcal{A}^1, \dots, \mathcal{A}^7$ 
			of marked sets (subsets of black intervals) which are used to ``mark'' those intervals which ``branch'' from the base interval line at depicted branching points.
			The marked sets belonging to a particular family $\ca A^i$ are depicted by the colored horizontal strips covering the black intervals belonging to those sets
			(e.g., the family $\ca A^1$ consists of three distinct sets, among which one occurs with multiplicity three), and each one forms a clique.}
		\label{fig:terminalSets}
	\end{figure}

	Even more informally, the previous can be illustrated by a picture in~\Cref{fig:terminalSets}, in which the interval
	graph represents a linear part of the clique-tree representation of a chordal graph, and the clique-cutsets separate
	the part from other ``branching'' parts of the clique-tree.

	\section{Interval Graphs and PQ-trees}\label{sec:intPQ}
	
	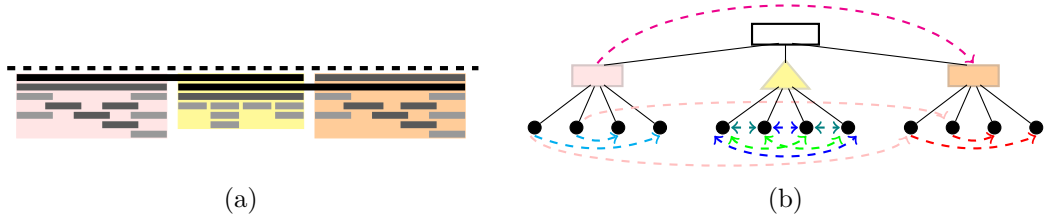
\begin{figure}[tb]
		\centering
		\begin{subfigure}[t]{0.47\linewidth}
			\centering
			\begin{tikzpicture}[xscale=0.5,yscale=-0.5]
				
				\draw[black, dashed, ultra thick] (-7,-4.5) -- (5.4,-4.5);
				
				\draw [draw=pink, fill=pink, opacity=0.4] (-6.75,-4.35) rectangle (-2.8,-2.65);
				
				\draw [draw=yellow, fill=yellow, opacity=0.4] (-2.5,-4.35) rectangle (0.8,-2.9);
				
				\draw [draw=orange, fill=orange, opacity=0.4] (1.1,-4.35) rectangle (5.05,-2.65);
				
				\draw[black, line width=2.6] (-6.75,-4.25) -- (0.8,-4.25); 
				
				\draw[black!65!white, line width=2.6] (-6.75,-4) -- (-2.8,-4); 
				
				\draw[black!40!white, line width =2.6] (-6.75,-3.75) -- (-5.8,-3.75); 
				\draw[black!40!white, line width =2.6] (-6.75,-3.25) -- (-5.8,-3.25); 
				\draw[black!65!white, line width =2.6] (-6,-3.5) -- (-5.05,-3.5); 
				\draw[black!65!white, line width =2.6] (-5.25,-3.25) -- (-4.3,-3.25); 
				\draw[black!65!white, line width =2.6] (-4.5,-3.5) -- (-3.55,-3.5); 
				\draw[black!65!white, line width =2.6] (-4.5,-3) -- (-3.55,-3); 
				\draw[black!40!white, line width =2.6] (-3.75,-3.75) -- (-2.8,-3.75); 
				\draw[black!40!white, line width =2.6] (-3.75,-3.25) -- (-2.8,-3.25); 
				\draw[black!40!white, line width =2.6] (-3.75,-2.75) -- (-2.8,-2.75); 
				
				\draw[black, line width=2.6] (-2.5,-4) -- (5.05,-4); 
				
				\draw[black!65!white, line width =2.6] (-2.5,-3.75) -- (0.8,-3.75); 
				
				\draw[black!40!white, line width =2.6] (-2.5,-3.5) -- (-1.75,-3.5); 
				\draw[black!40!white, line width =2.6] (-1.65,-3.5) -- (-0.9,-3.5); 
				\draw[black!40!white, line width =2.6] (-0.8,-3.5) -- (-0.05,-3.5); 
				\draw[black!40!white, line width =2.6] (0.05,-3.5) -- (0.8,-3.5); 
				
				\draw[black!40!white, line width =2.6] (-1.65,-3) -- (-0.9,-3); 
				\draw[black!40!white, line width =2.6] (-1.65,-3.25) -- (-0.9,-3.25); 
				\draw[black!40!white, line width =2.6] (0.05,-3.25) -- (0.8,-3.25); 
				
				\draw[black!40!white, line width =2.6] (1.1,-3.75) -- (2.05,-3.75); 
				\draw[black!40!white, line width =2.6] (1.1,-3.25) -- (2.05,-3.25); 
				\draw[black!65!white, line width =2.6] (1.85,-3.5) -- (2.9,-3.5); 
				\draw[black!65!white, line width =2.6] (2.6,-3.25) -- (3.55,-3.25); 
				\draw[black!65!white, line width =2.6] (3.35,-3.5) -- (4.3,-3.5); 
				\draw[black!65!white, line width =2.6] (3.35,-3) -- (4.3,-3); 
				\draw[black!40!white, line width =2.6] (4.1,-3.75) -- (5.05,-3.75); 
				\draw[black!40!white, line width =2.6] (4.1,-3.25) -- (5.05,-3.25); 
				\draw[black!40!white, line width =2.6] (4.1,-2.75) -- (5.05,-2.75); 
				
				\draw[black!65!white, line width=2.6] (1.1,-4.25) -- (5.05,-4.25); 
				
				\node (11) at (-0.85,-1) {(a)};	\label{a}
			\end{tikzpicture}
			
			\label{fig:intervalG}
		\end{subfigure}
		\hfill
		\begin{subfigure}[t]{0.51\linewidth}
			\centering
			\begin{tikzpicture}[yscale=0.55,xscale=1.1]
				
				\draw [draw=black, fill=white, thick, opacity=1] (1.1,-2.5) rectangle (1.9,-2);
				
				\draw [draw=pink!60!black, fill=pink, thick, opacity=0.4] (-1.05,-3.5) rectangle (-0.45,-3); 
				
				\node[isosceles triangle,
				isosceles triangle apex angle=80,
				draw,
				rotate=90,
				color=yellow!60!black, 
				fill=yellow, 
				thick,
				opacity=0.4,
				minimum size =0.2cm] (T1)at (1.5,-3.35){};
				
				\draw [draw=orange!60!black, fill=orange, thick, opacity=0.4] (3.45,-3.5) rectangle (4.05,-3); 
				
				\node at (1.5,-2.45) [draw, fill=black, opacity=0, color=black, inner sep=0.8mm] (a) {};
				\node at (-0.75,-3.015) [draw, fill=black, opacity=0, color=black, inner sep=0.8mm] (b) {};
				\draw (a) -- (b) node[midway, above, thick]{};
				
				\node at (3.75,-3.015) [draw, fill=black, opacity=0, color=black, inner sep=0.8mm] (c) {};
				\draw (a) -- (c) node[midway, above, thick]{};
				
				\node at (1.5,-2.35) [draw, fill=black, opacity=0, color=black, inner sep=0.8mm] (d) {};
				\node at (1.5,-3.1) [draw, fill=black, opacity=0, color=black, inner sep=0.8mm] (e) {};
				\draw (d) -- (e) node[midway, above, thick]{};

				\node at (-0.75,-3.35) [draw, fill=black, opacity=0, color=black, inner sep=0.8mm] (b1) {};
				\node at (-1.5,-4.5) [circle, draw, fill=black, opacity=1, color=black, inner sep=0.6mm] (b11) {};
				\draw (b1) -- (b11) node[midway, above, thick]{};
				
				\node at (-1,-4.5) [circle, draw, fill=black, opacity=1, color=black, inner sep=0.6mm] (b12) {};
				\draw (b1) -- (b12) node[midway, above, thick]{};
				
				\node at (-0.5,-4.5) [circle, draw, fill=black, opacity=1, color=black, inner sep=0.6mm] (b13) {};
				\draw (b1) -- (b13) node[midway, above, thick]{};
				
				\node at (0,-4.5) [circle, draw, fill=black, opacity=1, color=black, inner sep=0.6mm] (b14) {};
				\draw (b1) -- (b14) node[midway, above, thick]{};

				\draw[->,color=cyan, thick,dashed] (-1.5,-4.7) .. controls (-1,-5.1) and (-0.5,-5.1) .. (0,-4.7);	
				\draw[->,color=cyan, thick,dashed] (-1,-4.7) .. controls (-0.8,-4.85) and (-0.7,-4.85) .. (-0.5,-4.7);
				
				\draw[->,color=pink, thick,dashed] (-1.55,-4.7) .. controls (-1.25,-5.65) and (2.75,-5.65) .. (2.95,-4.7);
				\draw[->,color=pink, thick,dashed] (-0.9,-4.3) .. controls (-0.7,-3.75) and (3.3,-3.75) .. (3.4,-4.3);
				
				\draw[<->,color=teal, thick,dashed] (0.85,-4.5) .. controls (1.1,-4.5) and (1.1,-4.5) .. (1.15,-4.5);	
				\draw[<->,color=teal, thick,dashed] (1.85,-4.5) .. controls (2.1,-4.5) and (2.1,-4.5) .. (2.15,-4.5);
				
				\draw[<->,color=green, thick,dashed] (0.85,-4.7) .. controls (0.95,-5.1) and (1.55,-5.1) .. (1.75,-4.7);
				\draw[<->,color=green, thick,dashed] (1.25,-4.7) .. controls (1.45,-5.1) and (2.05,-5.1) .. (2.15,-4.7);
				
				\draw[<->,color=blue, thick,dashed] (0.65,-4.7) .. controls (0.95,-5.35) and (2.05,-5.35) .. (2.35,-4.7);
				\draw[<->,color=blue, thick,dashed] (1.35,-4.5) .. controls (1.5,-4.5) and (1.5,-4.5) .. (1.65,-4.5);
				
				\draw[->,color=red, thick,dashed] (3,-4.7) .. controls (3.5,-5.1) and (4,-5.1) .. (4.5,-4.7);	
				\draw[->,color=red, thick,dashed] (3.5,-4.7) .. controls (3.7,-4.85) and (3.8,-4.85) .. (4,-4.7);
				
				\draw[->,color=magenta, thick,dashed] (-0.75,-2.95) .. controls (0,-1.1) and (3,-1.1) .. (3.75,-2.95);	
				
				\node at (1.5,-3.4) [draw, fill=black, opacity=0, color=black, inner sep=0.8mm] (b2) {};
				\node at (0.75,-4.5) [circle, draw, fill=black, opacity=1, color=black, inner sep=0.6mm] (b21) {};
				\draw (b2) -- (b21) node[midway, above, thick]{};
				
				\node at (1.25,-4.5) [circle, draw, fill=black, opacity=1, color=black, inner sep=0.6mm] (b22) {};
				\draw (b2) -- (b22) node[midway, above, thick]{};
				
				\node at (1.75,-4.5) [circle, draw, fill=black, opacity=1, color=black, inner sep=0.6mm] (b23) {};
				\draw (b2) -- (b23) node[midway, above, thick]{};
				
				\node at (2.25,-4.5) [circle, draw, fill=black, opacity=1, color=black, inner sep=0.6mm] (b24) {};
				\draw (b2) -- (b24) node[midway, above, thick]{};
				
				\node at (3.75,-3.35) [draw, fill=black, opacity=0, color=black, inner sep=0.8mm] (b3) {};
				\node at (3,-4.5) [circle, draw, fill=black, opacity=1, color=black, inner sep=0.6mm] (b31) {};
				\draw (b3) -- (b31) node[midway, above, thick]{};
				
				\node at (3.5,-4.5) [circle, draw, fill=black, opacity=1, color=black, inner sep=0.6mm] (b32) {};
				\draw (b3) -- (b32) node[midway, above, thick]{};
				
				\node at (4,-4.5) [circle, draw, fill=black, opacity=1, color=black, inner sep=0.6mm] (b33) {};
				\draw (b3) -- (b33) node[midway, above, thick]{};
				
				\node at (4.5,-4.5) [circle, draw, fill=black, opacity=1, color=black, inner sep=0.6mm] (b34) {};
				\draw (b3) -- (b34) node[midway, above, thick]{};
				
				\node (11) at (1.5,-6.25) {(b)};	\label{b}
			\end{tikzpicture}
			
			\label{fig:PQ-possible}
		\end{subfigure}
		\caption{(a) An interval representation of a graph $G$, and (b) its unique (up to equivalence transformations) PQ-tree.
			P-nodes are triangle-shaped and Q-nodes are rectangle-shaped. Note also the three levels of gray used to depict
			the intervals in (a) that determine to which level of nodes in (b) the intervals are assigned as inner vertices (cf.~\Cref{sec:intPQ}).}
		\label{fig:PQ-tree}
	\end{figure}

	A \emph{clique} in a graph is a set of its vertices which are pairwise adjacent. A clique is \emph{maximal} if it can not be extended by adding another vertex. 
	Interval graphs have linearly many maximal cliques which can easily be listed in linear time using the simplicial-vertex elimination procedure (a simplicial vertex is one whose neighbors induce a clique). 
	
	To recognize and test the isomorphism of interval graphs, Booth and Lueker \cite{recogIntervalLinear} invented \emph{PQ-trees} (see \Cref{fig:PQ-tree}), 
	ordered rooted trees which have the maximal cliques of an interval graph in its leaves, and every internal node is either one of the following: 
	\begin{itemize}
		\item A \emph{P-node}: the order of its children can be permuted arbitrarily.
		\item A \emph{Q-node} : the order of its children can be reversed (but not changed otherwise).
	\end{itemize}
	The above permissible reorderings at P- and Q-nodes are called \emph{equivalence transformations} of a PQ-tree, and the following is well-known:
	
	\begin{theorem}[Booth and Lueker \cite{recogIntervalLinear}]\label{thm:recogIntervalLinear}
		For every interval graph $G$, one can in linear time construct a PQ-tree~$T$, such that the following hold:
		\begin{itemize}
			\item This PQ-tree $T$ of $G$ is unique up to equivalence transformations.
			\item Every possible interval representation of $G$ corresponds to a PQ-tree $T'$ of $G$ (i.e., one equivalent to~$T$).
			The correspondence is that the linear order of the maximal cliques in the representation of $G$ is the same as the one given by the linear order of the leaves of~$T'$.
		\end{itemize}
	\end{theorem}
	
	In particular, the latter point means that every automorphism of $G$ can be represented as an equivalence transformation of a PQ-tree of~$G$ (though, not the other way round without further information associated with the tree).
	One may go further this way. We say that an assignment of PQ-trees to interval graphs is {\em canonical} if, whenever we take isomorphic graphs $G\simeq G'$ and their canonical PQ-trees $T$ and $T'$, then $T$ and $T'$ are isomorphic respecting the order of the trees (one may say ``the same'').
	The following fact is crucial for us:
	\begin{corollary}[Colbourn and Booth \cite{AutMPQtrees}, noted already in~\cite{recogIntervalLinear}]\label{cor:canonint}
		For every interval graph $G$, one can in linear time compute the automorphism group of~$G$ and a PQ-tree of $G$ which is canonical.
	\end{corollary}
	
	The definition of a PQ-tree (of an interval graph) explicitly refers only to the maximal cliques of $G$ and not directly to its vertices,
	but it will be useful to clearly understand the relation of PQ-tree nodes to the particular vertices of $G$.
	Every node $p$ of a PQ-tree $T$ of $G$ can be associated with a subgraph of $G$ formed by the union of all cliques of the descendant leaves of~$p$ -- this subgraph is said to {\em belong} to~$p$.
	Then, for a node~$p$, we define the {\em inner vertices assigned to~$p$} as those vertices of $G$ which belong to $p$ and, if $p$ is not a leaf, they belong to at least two child nodes of $p$,
	but they do not belong to any node which is not an ancestor or a descendant of $p$. (See the illustration in \Cref{fig:PQ-tree}, and further in \Cref{fig:lemma13}.)

	Note that every vertex of $G$ is an inner vertex of precisely one node of its PQ-tree~$T$.
	Moreover, by the `consecutive-ones' property of a PQ-tree, the following holds;
	if a vertex $v$ of $G$ is an inner vertex of a node $p$ of $T$, and $v$ also belongs to a son $p_1$ of~$p$, then $v$ belongs to every descendant of $p_1$ in~$T$.
	This illustrates the unique nature of the node $p$ (to which $v$ is an inner vertex) for the vertex $v$ within the PQ-tree~$T$.
	
	In the case of a P-node, the inner vertices assigned to $p$ belong to all child nodes of $p$, but this is generally not true for Q-nodes.
	We thus additionally define the {\em ranking of inner vertices}; this is trivial for P-nodes (all inner vertices of the same rank).
	For a Q-node $q$ of $T$, we index the sons of $q$ from left o right in a palindromic way (i.e., as $1,2,3,2,1$ or $1,2,3,3,2,1$ depending on parity), which is invariant upon reversal.
	The {\em rank} of every inner vertex $w$ assigned to $q$ is then the multiset of indices of the sons of $q$ that $w$ belongs to (obviously, this must be a consecutive section of the index sequence).
	Observe that since two inner vertices of the same node and of the same rank are in the same collection of maximal cliques of $G$, they are mutually symmetric in the automorphism group of~$G$.

	\section{Automorphisms of Set Families of Bounded Antichain Size}\label{sec:autsets}
	
	In this section, we give the main technical tool of this paper -- a procedure efficiently computing the automorphism group of a set family
	under the assumption of bounded antichain size, which builds on ideas used already in our past paper~\cite{aaolu2019isomorphism}.
	Here we formulate those ideas in an extended form as the standalone result in \Cref{thm:automsetfam}.
	Again, to stay on the more general side, we consider set families in the multiset setting, meaning that the same set may occur in a family multiple times
	(but this does not pose any additional difficulties in the coming arguments besides having to observe the multiplicity as a label on a set).
	
	Note that the problem of computing the automorphism group of a colored set family is actually a special case of the problem
	{\sc AutomMarkedINT}$(G;\,\ca A^1,\ldots,\ca A^m)$ where $G$ is a clique, but, at the same time,
	we are going to prove in the next \Cref{sec:setstoint} that the {\sc AutomMarkedINT} reduces to the problem solved here.
	
	\begin{definition}[Automorphism group of a colored set family]\label{def:automsetfam}~\rm\\
		The problem {\sc AutomSET}$(X;\,\ca U^1,\ldots,\ca U^m)$ is defined as follows:
		\begin{description}
			\item[Input] For a finite ground set $X$, a finite set family $\ca U\subseteq2^X$ partitioned into $m\geq1$ color classes $\ca U=\ca U^1\cup\ldots\cup\ca U^m$
			(allowing the same set to occur in a family multiple times).
			\item[Task] Compute the group $\Gamma$ of the color-preserving permutations of $\ca U$ which come from a permutation of the ground set~$X$.
			Precisely, a permutation $\tau$ of $\ca U$ belongs to $\Gamma$, if and only if there exists a permutation $\sigma$ of~$X$
			such that, for every $i\in\{1,\ldots,m\}$ and all $B\in\ca U^i$, we have~$\tau(B)\in\ca U^i$ and~$\tau(B)=\sigma(B)$
			(this trivially gives $|B|=|\tau(B)|$).
			Note that the latter condition also immediately implies that the (possible) multiplicities of~$B$~and~of~$\sigma(B)$ in each family $\ca U^i$ are equal.
		\end{description}
		\smallskip
		The problem {\sc AutomSimpleSET}$(X;\,\ca U^1,\ldots,\ca U^m)$ is the same as {\sc AutomSET}$(X;\,\ca U^1,\ldots,\ca U^m)$ with the following condition on the input:
		For every set $B\in\ca U$, there is exactly one index $i\in\{1,\ldots,m\}$ such that $B\in\ca U^i$, and the multiplicity of $B$ in $\ca U^i$ equals one.
	\end{definition}
	
	We again, as with \Cref{thm:autommarked} above, estimate the input size here by the same simple argument.
	If the maximum antichain size in the family $\ca U=\ca U^1\cup\ldots\cup\ca U^m\subseteq2^X$ equals~$a$, then the input size of an instance 
	of {\sc AutomSET}$(X;\,\ca U^1,\ldots,\ca U^m)$ is at most $\sum_{U\in\ca U}|U| \in \ca O(a\cdot|X|^2)$.
	(Possible multiplicities of sets in the family $\ca U$ are negligible in this regard since they are encoded as integer labels of the multiple sets.)
	
	We start with a simple observation that will simplify the next \Cref{thm:automsetfam}:
	\begin{proposition2rep}\apxmark\label{pro:automsetfam-s}
		The problem {\sc AutomSET}$(X;\,\ca V^1,\ldots,\ca V^n)$ reduces in linear time to the problem {\sc AutomSimpleSET}$(X;\,\ca U^1,\ldots,\ca U^m)$
		for suitable~$\ca U^1,\ldots,\ca U^m$.
	\end{proposition2rep}
	\begin{proof}
		For every set $B\in\ca V=\ca V^1\cup\ldots\cup\ca V^n$, we record the integer vector $m_B:=(b_1,\ldots,b_n)$ where $b_i$ is the multiplicity of $B$ in the family $\ca V^i$.
		As noted already in Definition~\ref{def:automsetfam}, any permutation $\tau$ in the solution of {\sc AutomSET}$(X;\,\ca V^1,\ldots,\ca V^n)$ must preserve this multiplicity vector; $m_B=m_{\tau(B)}$.
		Let $\ca U$ be the simplification of the (multi)family $\ca V$, i.e., without repetition of multiple member sets.
		We hence define a partition $(\ca U^1,\ldots,\ca U^m)$ as the partition of $\ca U$ given by the equality of the vectors~$m_B$.
		Then $\tau$ projects to a permutation in {\sc AutomSimpleSET}$(X;\,\ca U^1,\ldots,\ca U^m)$.
		
		Conversely, for any permutation $\sigma$ in the solution of {\sc AutomSimpleSET}$(X;\,\ca U^1,\ldots,\ca U^m)$, we expand $\sigma$ to permutations $\tau$ of $\ca V$ as follows.
		Every functional assignment $B\mapsto\sigma(B)$ is lifted, for each $i\in\{1,\ldots,m\}$, to all bijections of the set (possibly empty) of multiple copies of $B$ in $\ca V^i$ to the set of multiple copies of $\sigma(B)$ in $\ca V^i$.
		Then every such $\tau$ belongs to the solution group of {\sc AutomSET}$(X;\,\ca U^1,\ldots,\ca U^m)$.
	\end{proof}

	On the other hand, we remark that in the problem {\sc AutomSimpleSET}$(X;\,\ca U^1,\ldots,\ca U^m)$, the number $m$ of colors is not bounded,
	and so there is (likely) no easy way to reduce this problem to the uncolored case.
	
	\begin{theorem}\label{thm:automsetfam}
		The problem {\sc AutomSET}$(X;\,\ca U^1,\ldots,\ca U^m)$ is solvable in FPT-time with respect to the parameter $a$ which is the maximum antichain size of $\ca U$.
	\end{theorem}
	
	The rest of the section is devoted to the proof of \Cref{thm:automsetfam} which, in view of Proposition~\ref{pro:automsetfam-s}, is sufficient to prove for the problem {\sc AutomSimpleSET}$(X;\,\ca U^1,\ldots,\ca U^m)$.
	
	\subparagraph{Cardinality Venn diagrams.}
	First of all, we review what Definition~\ref{def:automsetfam}, specifically the words ``there exists a permutation $\sigma$ of~$X$'', mean for us.
	Since it is not much efficient to deal with the many permutations of~$X$, we now show a simple fact that it will be enough to observe certain cardinalities to decide the existence of such permutation $\sigma$ as required in the definition.
	For a set family $\mathcal{U}$, we call a {\em cardinality Venn diagram} of $\mathcal{U}$ the integer vector
	$\big(\ell_{\mathcal{U},\mathcal{U}_1}: \emptyset\not=\mathcal{U}_1\subseteq\mathcal{U}\big)$
	such that $\ell_{\mathcal{U},\mathcal{U}_1}:=|L_{\mathcal{U},\mathcal{U}_1}|$ where
	$L_{\mathcal{U},\mathcal{U}_1}=\bigcap_{A\in\mathcal{U}_1}\!A \setminus \bigcup_{B\in{\mathcal{U}\setminus\mathcal{U}_1}}\!B$.
	\vspace{1pt}%
	That is, informally, we record the cardinality of every internal cell of the Venn diagram of~$\mathcal{U}$.
	See an illustration in Figure~\ref{fig:vennIL}.
	
	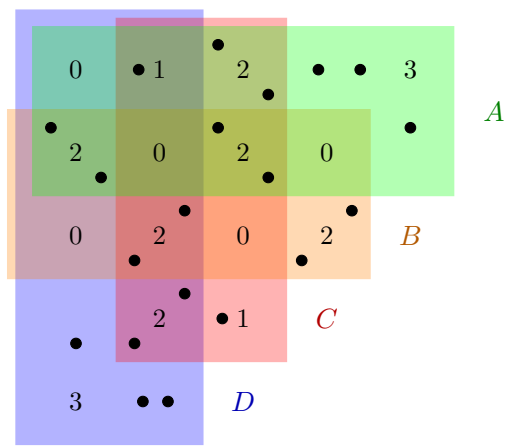
\begin{figure}[tb]
		\centering
		\begin{tikzpicture}[scale=1.1]
			
			\filldraw [fill=blue, draw=blue, ultra thick, opacity=0.3] (-0.2,5.2) rectangle (2,0); 
			
			\filldraw [fill=red, draw=red, ultra thick, opacity=0.3] (1,5.1) rectangle (3,1); 
			
			\filldraw [fill=orange, draw=orange, ultra thick, opacity=0.3] (-0.3,4) rectangle (4,2); 
			
			\filldraw [fill=green, draw=green, ultra thick, opacity=0.3] (0,5) rectangle (5,3); 
			
			\node (1) at (0.5,0.5) {$3$};
			\node (2) at (2.5,1.5) {$1$};
			\node (3) at (3.5,2.5) {$2$};
			\node (4) at (4.5,4.5) {$3$};
			
			\node (12) at (1.5,1.5) {$2$};
			\node (13) at (0.5,2.5) {$0$};
			\node (14) at (0.5,4.5) {$0$};
			\node (23) at (2.5,2.5) {$0$};
			\node (24) at (2.5,4.5) {$2$};
			\node (34) at (3.5,3.5) {$0$};
			
			\node (123) at (1.5,2.5) {$2$};
			\node (124) at (1.5,4.5) {$1$};
			\node (134) at (0.5,3.5) {$2$};
			\node (234) at (2.5,3.5) {$2$};
			
			\node (1234) at (1.5,3.5) {$0$};
			
			\node (A) at (5.5,4) {\color{green!50!black}\large$A$};
			\node (B) at (4.5,2.5) {\color{orange!70!black}\large$B$};
			\node (C) at (3.5,1.5) {\color{red!70!black}\large$C$};
			\node (D) at (2.5,0.5) {\color{blue!70!black}\large$D$};
			
			\tikzstyle{every node}=[draw, shape=circle, inner sep=1.4pt, fill=black]
			\node at (0.5,1.2) {};	\node at (1.3,0.5) {};	\node at (1.6,0.5) {};
			\node at (3.4,4.5) {};	\node at (3.9,4.5) {};	\node at (4.5,3.8) {};
			\node at (1.2,1.2) {};	\node at (1.8,1.8) {};
			\node at (2.25,1.5) {};
			\node at (1.2,2.2) {};	\node at (1.8,2.8) {};
			\node at (0.8,3.2) {};	\node at (0.2,3.8) {};
			\node at (1.25,4.5) {};
			\node at (2.8,3.2) {};	\node at (2.2,3.8) {};
			\node at (2.8,4.2) {};	\node at (2.2,4.8) {};
			\node at (3.2,2.2) {};	\node at (3.8,2.8) {};
		\end{tikzpicture}
		\caption{An illustration of the concept of a cardinality Venn diagram, and of Lemma~\ref{lem:preciseVenn}.
			We have $4$ sets $A,B,C,D$, and the ground set elements are depicted with the dots. The numbers in the cells of the diagram are the cardinalities of these cells.
			One can check, using the depicted cardinalities, that there exists a permutation of the ground set which permutes our sets $(A,B,C,D)$ into $(D,B,C,A)$,
			but there is no such permutation permuting $(A,B,C,D)$ into $(A,C,B,D)$.}
		\label{fig:vennIL}
	\end{figure}

	For $\ca U_1\subseteq\ca U$, let naturally $\varrho(\mathcal{U}_1)=\{\varrho(B):B\in\mathcal{U}_1\}$.
	We have easily got:
	
	\begin{lemma2rep}[\cite{aaolu2019isomorphism}]\apxmark\label{lem:preciseVenn}
		Let $\varrho$ be a permutation of~$\ca U$ over~$X$.
		There exists a permutation $\sigma$ of~$X$ such that, for every $B\in\ca U$, we have~$\varrho(B)=\sigma(B)$, if and only if
		the cardinality Venn diagrams of $\mathcal{U}$ and of $\varrho(\mathcal{U})$ are the same (equal), meaning that
		$\ell_{\mathcal{U},\mathcal{U}_1}= \ell_{\mathcal{U},\varrho(\mathcal{U}_1)}$~for~all~$\emptyset\not=\mathcal{U}_1\subseteq\mathcal{U}$.
		
		Furthermore, for $\mathcal{U'}\subseteq\mathcal{U}$ such that $\varrho(\mathcal{U'})=\mathcal{U'}$,
		one can in $\ca O\big(|X|+\sum_{U\in\ca U'}|U|\big)$ time test the condition as above, i.e., whether the equalities
		$\ell_{\mathcal{U}'\!,\,\mathcal{U}_1}= \ell_{\mathcal{U}'\!,\,\varrho(\mathcal{U}_1)}$ hold for all~$\emptyset\not=\mathcal{U}_1\subseteq\mathcal{U}'$.
	\end{lemma2rep}
	
	\begin{proof}
		$\Rightarrow$ 	Suppose that there exists such a permutation $\sigma$ of~$X$.
		Then, for all~$\emptyset\not=\mathcal{U}_1\subseteq\mathcal{U}$, every element of $L_{\mathcal{U},\mathcal{U}_1}$ is mapped by $\sigma$ into $L_{\mathcal{U},\varrho(\mathcal{U}_1)}$,
		and so the claim follows since $\sigma$ is a permutation.
		
		$\Leftarrow$	For any~$\emptyset\not=\mathcal{U}_1\subseteq\mathcal{U}$, we have that $|L_{\mathcal{U},\mathcal{U}_1}|=|L_{\mathcal{U},\varrho(\mathcal{U}_1)}|$
		which implies an existence of a bijection from $L_{\mathcal{U},\mathcal{U}_1}$ to $L_{\mathcal{U},\varrho(\mathcal{U}_1)}$.
		Since $L_{\mathcal{U},\mathcal{U}_1}\cap L_{\mathcal{U},\mathcal{U}_2}=\emptyset$ for $\ca U_1\not=\ca U_2$, the composition of these bijections is sound and results in a permutation $\sigma$ of~$X$.
		Picking any $B\in\ca U$, we have $B=\bigcup\big\{L_{\mathcal{U},\mathcal{U}_1}: \{B\}\subseteq\mathcal{U}_1\subseteq\mathcal{U}\big\}$,
		and hence $\sigma(B)=\bigcup\big\{\sigma(L_{\mathcal{U},\mathcal{U}_1}): \{B\}\subseteq\mathcal{U}_1\subseteq\mathcal{U}\big\}= \bigcup\big\{L_{\mathcal{U},\varrho(\mathcal{U}_1)}: \{B\}\subseteq\mathcal{U}_1\subseteq\mathcal{U}\big\}= \varrho(B)$.
		
		\smallskip
		As for testing the condition `$\ell_{\mathcal{U}'\!,\,\mathcal{U}_1}= \ell_{\mathcal{U}'\!,\,\varrho(\mathcal{U}_1)}$ hold for all~$\emptyset\not=\mathcal{U}_1\subseteq\mathcal{U}'$\,',
		we loop through all elements $x\in X$, and for each $x$ we record in $\mathcal{O}(n)$ time to which of the sets in $\mathcal{U'}$ this $x$ belongs to.
		Summing the obtained records at the~end precisely gives the $\mathcal{O}(|X|)$ nonzero values
		$\ell_{\mathcal{U}'\!,\,\mathcal{U}_1}$ over $\emptyset\not=\mathcal{U}_1\subseteq\mathcal{U}'$.
		We analogously compute the values $\ell_{\mathcal{U}'\!,\,\varrho(\mathcal{U}_1)}$ over $\emptyset\not=\mathcal{U}_1\subseteq\mathcal{U}'$, and then compare.
	\end{proof}
	
	Our strategy for the proof of \Cref{thm:automsetfam} is as follows.
	\begin{itemize}
		\item Let, for $n=|X|$ and every $j\in\{1,\ldots,m\}$,~ $\ca U^j=\bigcup_{c=1}^n\ca U^j_c$ be a partition of $\ca U^j$ into subfamilies of sets of cardinality $c\in\{1,\ldots,n\}$.
		We make the initial group $\Gamma'$ (of permu\-tations of~$\ca U=\bigcup_{j=1}^m\ca U^j$) as the direct product of the symmetric groups (those of all permutations) on all nonempty families $\ca U^j_c$ over $j\in\{1,\ldots,m\}$ and $c\in\{1,\ldots,n\}$.
		\item We compute the subgroup $\Gamma\subseteq\Gamma'$ of those permutations which fulfill the condition `$\ell_{\mathcal{U},\mathcal{U}_1}= \ell_{\mathcal{U},\varrho(\mathcal{U}_1)}$ for all~$\emptyset\not=\mathcal{U}_1\subseteq\mathcal{U}$\,'
		of Lemma~\ref{lem:preciseVenn}. Then $\Gamma$ will be, by Lemma~\ref{lem:preciseVenn}, the solution to the {\sc AutomSimpleSET}$(X;\,\ca U^1,\ldots,\ca U^m)$ problem.
		At this point it will be quite important that each of the subfamilies $\ca U^j_c$ (as an antichain) has bounded cardinality.
	\end{itemize}
	\smallskip
	The latter point, however, is not an easy task, and we will employ classical Babai's tower-of-groups procedure to ``gradually refine'' $\Gamma'$ into $\Gamma$,
	ensuring that the conditions of Lemma~\ref{lem:preciseVenn} hold, in a sense, for more and more combinations of the subfamilies~$\ca U^j_c$ until all are satisfied.
	
	Before getting into the group-computing tools, we introduce one more technical result which will be crucial in the gradual refinement of $\Gamma'$ into $\Gamma$.
	In the setting of Lemma~\ref{lem:preciseVenn}, we say that $\mathcal{U}'\subseteq\mathcal{U}$ is {\em Venn-good} with $\varrho$ if
	$\ell_{\mathcal{U}'\!,\,\mathcal{U}_1} = \ell_{\varrho(\mathcal{U}'),\varrho(\mathcal{U}_1)}$ holds true for all $\emptyset\not=\mathcal{U}_1\subseteq\mathcal{U}'$,
	and we call such $\mathcal{U}_1$ a {\em witness} (of $\mathcal{U}'$ not being Venn-good) if $\ell_{\mathcal{U}'\!,\,\mathcal{U}_1} \not= \ell_{\varrho(\mathcal{U}'),\varrho(\mathcal{U}_1)}$.
	
	\begin{lemma2rep}[\cite{aaolu2019isomorphism}]\apxmark\label{lem:Venngood}
		Let $\varrho$ be a permutation of a set family~$\ca U$, and $\mathcal{U'}\subseteq\mathcal{U}$ be such that $\varrho(\mathcal{U'})=\mathcal{U'}$.
		If $\mathcal{U}'$ is not Venn-good with~$\varrho$, then there exist $\mathcal{U}_2,\mathcal{U}_3\subseteq\mathcal{U}'$ such that $|\mathcal{U}_2|\leq2$ or $\mathcal{U}_2$ is an antichain in the inclusion,
		$\emptyset\not=\mathcal{U}_3\subseteq\mathcal{U}_2$ and	$\ell_{\mathcal{U}_2,\mathcal{U}_3}\not= \ell_{\varrho(\mathcal{U}_2),\varrho(\mathcal{U}_3)}$
		(not~Venn-good).
	\end{lemma2rep}
	
	\begin{proof}
		Choose $\mathcal{U}_2\subseteq\mathcal{U}'$ such that $\mathcal{U}_2$ is not Venn-good with $\varrho$ and it is minimal such by inclusion,
		and assume (for a contradiction) that there are $A_1,A_2\in\mathcal{U}_2$ such that~$A_1\subseteq A_2$.
		If $\varrho(A_1)\not\subseteq\varrho(A_2)$, then already $\mathcal{U}_2:=\{A_1,A_2\}$ is not Venn-good (with a witness $\{A_1\}$), and so let $\varrho(A_1)\subseteq\varrho(A_2)$.
		Let $\mathcal{U}_3$ be a witness of $\mathcal{U}_2$ not being Venn-good, and for $j=2,3$ denote: 
		$\mathcal{V}_j^0:=\mathcal{U}_j\setminus\{A_1,A_2\}$,
		$\mathcal{V}_j^1:=\big(\mathcal{U}_j\cup\{A_1\}\big)\setminus\{A_2\}$,
		$\mathcal{V}_j^2:=\big(\mathcal{U}_j\cup\{A_2\}\big)\setminus\{A_1\}$,
		$\mathcal{V}_j^3:=\mathcal{U}_j\cup\{A_1,A_2\}$.
		
		Since $A_1\subseteq A_2$ and $\varrho(A_1)\subseteq\varrho(A_2)$ (and so $L_{\mathcal{U}_2,\mathcal{V}_3^1}=\emptyset$), we easily derive
		\begin{eqnarray*}
			\ell_{\mathcal{U}_2,\mathcal{V}_3^1} = &0& = \ell_{\varrho(\mathcal{U}_2),\varrho(\mathcal{V}_3^1)}
			\,,\end{eqnarray*}
		and since, by our minimality assumption, all three subfamilies $\mathcal{V}_2^0$, $\mathcal{V}_2^1$ and $\mathcal{V}_2^2$ are Venn-good,
		\begin{eqnarray*}
			\ell_{\mathcal{U}_2,\mathcal{V}_3^0} = 
			\ell_{\mathcal{U}_2\setminus\{\!A_1\!\},\,\mathcal{V}_3^0} =
			\ell_{\mathcal{V}_2^2,\,\mathcal{V}_3^0} 
			&\!=\!&
			\ell_{\varrho(\mathcal{V}_2^2),\varrho(\mathcal{V}_3^0)} =
			\ell_{\varrho(\mathcal{U}_2)\setminus\{\varrho(\!A_1\!)\},\varrho(\mathcal{V}_3^0)} =
			\ell_{\varrho(\mathcal{U}_2),\varrho(\mathcal{V}_3^0)}
			\,,\\
			\ell_{\mathcal{U}_2,\mathcal{V}_3^3} = 
			\ell_{\mathcal{U}_2\setminus\{\!A_2\!\},\,\mathcal{V}_3^3\setminus\{\!A_2\!\}} 
			=\ell_{\mathcal{V}_2^1,\,\mathcal{V}_3^1} 
			&\!=\!&
			\ell_{\varrho(\mathcal{V}_2^1),\varrho(\mathcal{V}_3^1)} =
			\ell_{\varrho(\mathcal{U}_2)\setminus\{\!\varrho(\!A_2\!)\!\},\varrho(\mathcal{V}_3^3)\setminus\{\!\varrho(\!A_2\!)\!\}} =
			\ell_{\varrho(\mathcal{U}_2),\varrho(\mathcal{V}_3^3)}
			\,.\end{eqnarray*}
		Then, using the previous equalities, and the trivial observation 
		$\ell_{\mathcal{V}_2^0,\,\mathcal{V}_3^0}= \ell_{\mathcal{U}_2,\mathcal{V}_3^0}+\ell_{\mathcal{U}_2,\mathcal{V}_3^1}+\ell_{\mathcal{U}_2,\mathcal{V}_3^2}+\ell_{\mathcal{U}_2,\mathcal{V}_3^3}$
		with the analogous equality under $\varrho$, we conclude
		\begin{eqnarray*}
			\ell_{\mathcal{U}_2,\mathcal{V}_3^2} &=&
			\ell_{\mathcal{V}_2^0,\,\mathcal{V}_3^0}
			-\ell_{\mathcal{U}_2,\mathcal{V}_3^0}-0
			-\ell_{\mathcal{U}_2,\mathcal{V}_3^3}
			\\	&=&\ell_{\varrho(\mathcal{V}_2^0),\varrho(\mathcal{V}_3^0)}
			-\ell_{\varrho(\mathcal{U}_2),\varrho(\mathcal{V}_3^0)} -0
			-\ell_{\varrho(\mathcal{U}_2),\varrho(\mathcal{V}_3^3)}
			=\>\ell_{\varrho(\mathcal{U}_2),\varrho(\mathcal{V}_3^2)}
			\,.\end{eqnarray*}
		However, $\mathcal{U}_3\in\{\mathcal{V}_3^0, \mathcal{V}_3^1,\mathcal{V}_3^2,\mathcal{V}_3^3\}$,
		and so one of the latter four derived equalities contradicts the assumption that $\mathcal{U}_3$ witnessed $\mathcal{U}_2$ not being Venn-good with~$\varrho$.
	\end{proof}
	
	Since we deal with set families of bounded-size antichains, Lemma~\ref{lem:Venngood} essentially tells us that either a family is Venn-good, or it contains a ``small'' subfamily not being Venn-good.

	\subparagraph{Babai's tower-of-groups~\cite{babai-bdcm}.}
	The famous classical paper of Babai~\cite{babai-bdcm} is not just a standalone algorithm, but more an outline of how to efficiently compute
	a subroup $\Gamma$ of a given group $\Gamma'$ in a setting in which the conditions defining the subgroup $\Gamma$ can be stepwise refined in sufficiently small steps.
	Here, by ``computing a~group'' we mean to output a set of its generators which is at most polynomially large by~\cite{furst}. 
	Note that this task, in general, cannot be done directly by processing all members (even though we have an efficient membership test at hand) of the group $\Gamma'$ which can be exponentially large.
	
	Babai's ``tower-of-groups'' approach informally works as follows: we iteratively compute a chain of subgroups
	$\Gamma'=\Gamma_0\supseteq\Gamma_1\supseteq\ldots\supseteq\Gamma_h=\Gamma$, where each $\Gamma_{i+1}$ consists of those members of $\Gamma_i$ which satisfy a suitably chosen additional condition,
	until $\Gamma_h=\Gamma$ satisfies all the defining conditions, and so it is the desired outcome.
	The important ingredient which makes this procedure work efficiently is that the ratio of orders (sizes) of consequent groups $\Gamma_i$ and $\Gamma_{i+1}$ in the chain is always bounded.
	The number $h$ of steps should also not be too large.
	In this setting, each of the refinement steps can be done using another classical result:
	
	\begin{theorem}[{Furst, Hopcroft and Luks \cite[Cor.~1]{furst}}]\label{thm:furstgen}~\\
		Let $\Pi$ be a permutation group given by its generators, and $\Pi_1$ be any subgroup of $\Pi$ such that one can test in
		polynomial time whether $\pi\in\Pi_1$ for any $\pi\in\Pi$ (membership test).
		If the ratio $|\Pi|/|\Pi_1|$ is bounded by a function of a parameter $d$, 
		then a set of generators of $\Pi_1$ can be computed in \textbf{FPT}-time (with respect to~$d$).
	\end{theorem}
	
	The collected ingredients show a clear road to computing the subgroup $\Gamma\subseteq\Gamma'$ from the above outlined proof strategy for \Cref{thm:automsetfam}.
	In every intermediate step of a chain $\Gamma'=\Gamma_0\supseteq\Gamma_1\supseteq\cdots\supseteq\Gamma_h=\Gamma$,
	we assume that the whole family $\ca U$ is not Venn-good (with some generator of $\Gamma_i$), and we use Lemma~\ref{lem:Venngood}
	to argue that there exists a small subfamily $\ca U_2$ of $\ca U$ which is not Venn-good with some generator of $\Gamma_{i-1}$.
	Such $\ca U_2$ can be found efficiently with a little trick (note that a brute-force approach would not give an FPT-time algorithm here), 
	and then we define and compute the next group $\Gamma_{i}$ as that of permutations for which $\ca U_2$ is Venn-good with \Cref{thm:furstgen}.
	In this way, we eventually arrive at the group $\Gamma$ such that $\ca U$ is Venn-good with every generator of~$\Gamma$.
	The details follow.

	\begin{proof}[Proof of \Cref{thm:automsetfam}]
		Recall that, for an instance of {\sc AutomSimpleSET}$(X;$ $\ca U^1,\ldots,\ca U^m)$ over an $n$-element set $X$ and for every $j\in\{1,\ldots,m\}$,
		we have defined a refined partition $\ca U^j=\bigcup_{c=1}^n\ca U^j_c$ where $\ca U^j_c$ consists of the sets from $\ca U^j$ of cardinality~$c$.
		Since the maximum antichain size of $\ca U$ is $a$, we have that each~$\ca U^j_c$ contains at most $a$ distinct sets.
		For simplicity, we may assume $a\geq2$ since the case of $a=1$ is trivial.
		Let $\ca W:=\{\ca U^j_c: 1\leq j\leq m, 1\leq c\leq n,\, \ca U^j_c\not=\emptyset\}$ be a system of all these nonempty families.
		
		\begin{algorithm}[tb]
			\caption{One ($i$-th) step of the computation of the subgroup $\Gamma\subseteq\Gamma'$}
			\label{alg:Gammaprime}
			\begin{algorithmic}[1]\smallskip
				\Require a set family $\mathcal{U}$ of maximum antichain size~$a\in\mathbb N$, a partition $\ca W\subseteq2^{\ca U}$ of $\ca U$ into parts of size~$\leq a$,
				and a group $\Gamma_{i-1}$ (via a generator set) of permutations of $\ca U$ which set-wise stabilizes every part in~$\ca W$.
				
				\Ensure either a certificate that $\mathcal{U}$ is Venn-good with every generator of $\Gamma_{i-1}$; or
				
				\noindent a subgroup $\Gamma_i\subsetneq\Gamma_{i-1}$ (via a generator set) such that, for $\mathcal{T}\subseteq\mathcal{U}$ which is the union of some at most $a$ parts of~$\ca W$,~
				$\mathcal{T}$ is Venn-good precisely with every member of $\Gamma_i$ (and not with members of $\Gamma_{i-1}\setminus\Gamma_i$).
				\smallskip
				
				\State Let $\ca W=\{\ca W_1,\ca W_2,\ldots,\ca W_k\}$
				\State $\ca T \leftarrow \emptyset$
				\Repeat{~for $p:=1,2,\ldots,a\,$}:
				\Repeat{~for $q:=1,2,\ldots,k+1\,$}:
				\If{$q>k$}
				\Return ``\,$\mathcal{U}$ is Venn-good with all generators of $\Gamma_{i-1}$'';
				\label{it:allgood}\EndIf{}
				\State $\ca T_1 \leftarrow $ $(\ca W_1\cup\ca W_2\cup\dots\cup\ca W_q)$ $\cup$ $\ca T$;
				\Until{ $\ca T_1$ is not Venn-good (cf.~Lemma~\ref{lem:preciseVenn}) with some generator of $\Gamma_{i-1}$};
				\State $j_{p} \leftarrow q$;
				\State $\ca T \leftarrow $ $\ca W_{j_1}\cup\ca W_{j_2}\cup\dots\cup\ca W_{j_{p}}$;
				
				\Until{$j_p=1$ or $p=a$ };\label{it:leavecyc}
				
				\State Call the algorithm of Theorem~\ref{thm:furstgen} to compute the subgroup $\Gamma_i\subseteq\Gamma_{i-1}$, such that the
				membership test of $\varrho\in\Gamma_i$ checks whether $\ca T$ is Venn-good with $\varrho$ (again~Lemma~\ref{lem:preciseVenn});
				\label{it:callFurst}
				\State\Return { $\Gamma_i$ }
			\end{algorithmic}
		\end{algorithm}
		
		To solve the given instance, we start with the group $\Gamma'=\Gamma_0$ of permutations of $\ca U$ which is the direct product of the symmetric groups on all parts of~$\ca W$.
		For $i:=1,2,\ldots$, we iteratively call \Cref{alg:Gammaprime} with $\ca U$, $\ca W$ and $\Gamma_{i-1}$, 
		until the outcome (in $h$-th step) is that $\ca U$ is Venn-good with all generators of~$\Gamma_h$.
		Note that the latter immediately gives that $\ca U$ is Venn-good with all permutations in~$\Gamma_h$.
		By the assurance of \Cref{alg:Gammaprime}, $\Gamma=\Gamma_h$ contains all permutations for which $\ca U$ is Venn-good
		and indeed is the desired solution of {\sc AutomSimpleSET}$(X;\,\ca U^1,\ldots,\ca U^m)$.
		
		So, it remains to finish two things; analyze and prove one call to Algorithm~\ref{alg:Gammaprime}, and prove that the number $h$ of steps is finite and not ``too large''.
		We start with the former. If $\mathcal{U}$ is Venn-good for every generator of $\Gamma_{i-1}$, then we find this already in the first iteration of $p=1$, on line \ref{it:allgood}.
		Hence we may further assume that $\mathcal{U}$ is not Venn-good.
		By Lemma~\ref{lem:Venngood}, there exists a subfamily $\ca U_2\subseteq\ca U$ which is also not Venn-good, and $|\ca U_2|\leq\max(a,2)\leq a$ since $\ca U_2$ is an antichain.
		Note that $\ca U_2$ thus intersects at most $a$ parts of $\ca W$, which implies that there exist at most $a$ parts of $\ca W$ whose union is not Venn-good.
		
		Let $k\geq j'_1>j'_2>\dots>j'_r\geq1$ be an index sequence of length $r\leq a$ such that the subfamily $\ca W_{j'_1}\cup\ca W_{j'_2}\cup\dots\cup\ca W_{j'_r}$ is not Venn-good with some generator of $\Gamma_{i-1}$,
		and the vector $(j'_1,j'_2,\ldots,j'_r)$ is lexicographically minimal of these properties.
		Then one can straightforwardly verify that $(j'_1,j'_2,\ldots,j'_r)$ is a prefix of (or equal to) the vector $(j_1,j_2,\ldots,j_p)$ computed by \Cref{alg:Gammaprime}.
		Consequently, the collection $\ca T$, when leaving the cycle on line \ref{it:leavecyc}, is not Venn-good with some generator of $\Gamma_{i-1}$.
		We look at the subset $\Gamma_i\subseteq\Gamma_{i-1}$ defined as on line~\ref{it:callFurst}. In particular, $\Gamma_i\subsetneq\Gamma_{i-1}$.
		The important point is that, as $\Gamma_{i-1}$ set-wise stabilizes every part of~$\ca W$,\, $\Gamma_i$ is closed under composition of permutations, and so it forms a subgroup as expected by the algorithm.
		
		Next, we verify the fulfillment of the assumptions of \Cref{thm:furstgen}.
		Generators of $\Gamma_{i-1}=\Pi$ have been given to \Cref{alg:Gammaprime}.
		The ratio $|\Pi|/|\Pi_1|$, where $\Pi_1=\Gamma_i$ in our case, can be bounded as follows (despite we do not know $\Gamma_i$ yet):
		by standard algebraic arguments, $|\Gamma_{i-1}|/|\Gamma_{i}|$ equals the number of distinct cosets of the subgroup $\Gamma_{i}$ in $\Gamma_{i-1}$.
		If we consider two automorphisms $\alpha,\beta\in\Gamma_{i-1}$ which are equal when restricted to $\ca T$ (recall that they set-wise stabilize $\ca T$), 
		then the automorphism $\alpha^{-1}\beta$ determines a permutation of $\mathcal{U}$ which is identical on $\ca T$ (so it is Venn-good with~$\alpha^{-1}\beta$), and hence $\alpha^{-1}\beta\in\Gamma_{i}$.
		The latter means that $\alpha$ and $\beta$ belong to the same coset of $\Gamma_{i}$, and consequently, the number of distinct cosets 
		is at most the number of distinct subpermutations on $\ca T$ possibly induced by $\Gamma_{i-1}$, that is at most $(a!)^{a}$
		(at most the symmetric group on each of at most $a$ parts of $\ca W$ which form~$\ca T$).
		Therefore, we can finish one iteration on line \ref{it:callFurst} in \textbf{FPT}-time with respect to~$d=a$ by \Cref{thm:furstgen}.
		
		Lastly, we estimate the number of steps $h$ in the refinement process $\Gamma'=\Gamma_0\supsetneq\Gamma_1\supsetneq\ldots\supsetneq\Gamma_h=\Gamma$.
		By Lagrange's group theorem, $|\Gamma_{i}|$ divides $|\Gamma_{i-1}|$, and since $\Gamma_{i-1}\not=\Gamma_i$, we have $|\Gamma_{i}|\leq\frac12|\Gamma_{i-1}|$.
		Hence the number of strict refinement steps in our chain of subgroups is $h\leq\log_2|\Gamma'|$.
		Since trivially $|\Gamma'|\leq(a!)^n$ (where~$n=|X|$), we get $h=\mathcal{O}(na\log a)$.
		Therefore, the overall computation of the resulting group $\Gamma$ of the problem instance of {\sc AutomSimpleSET}$(X;\,\ca U^1,\ldots,\ca U^m)$
		is finished in FPT-time with respect to the parameter $a$, the maximum antichain size of the set family~$\ca U$.
		We remark that the only steps in our algorithm which require FPT-time (i.e., are possibly not of polynomial-time) are the calls to the algorithm of \Cref{thm:furstgen}.
	\end{proof}

	\section{Handling PQ-trees of marked interval graphs}\label{sec:setstoint}

	In this section, we provide the proof -- an FPT algorithm, for \Cref{thm:autommarked}.
	For this purpose, we revisit the PQ-trees of \Cref{sec:intPQ} from the point of view of marked interval graphs of Definition~\ref{def:autommarked}.
	While efficient handling of PQ-trees with an arbitrary marking seems basically infeasible, we deal with the special assumption of bounded antichains which appears very helpful within PQ-trees.
	Informally, we can say that the marked sets under this assumptions affect only very small part of the whole PQ-tree of our marked graph~$G$.

	Recall that we have got an interval graph $G$, families $\ca A^1,\ldots,\ca A^m$ of nonempty (marked) subsets of $V(G)$ 
	such that every set $A\in\ca A^i$ where $i\in\{1,\ldots,m\}$ induces a clique of~$G$, and $\ca A:=\ca A^1\cup\ldots\cup\ca A^m$.
	Let the maximum antichain size in $\ca A$ be~$a$.
	Let $T$ be a PQ-tree of~$G$, and recall what are inner vertices of $G$ assigned to nodes of~$T$ and their rank from Section~\ref{sec:intPQ}.
	We call a node $p$ of $T$ {\em clean} if the inner vertices assigned to $p$ are disjoint from $\bigcup\ca A$ (that is, not belonging to any marked set).
	The subtree rooted at $p$ is then {\em clean} if $p$ and all descendants of $p$ in $T$ are clean.
	It can be shown that the number of ``incomparable'' non-clean subtrees is bounded by the antichain size, but we skip the partial details since we actually prove much more in Lemma~\ref{lem:markedtosets}.
	See an illustration in~\Cref{fig:lemma13}.

	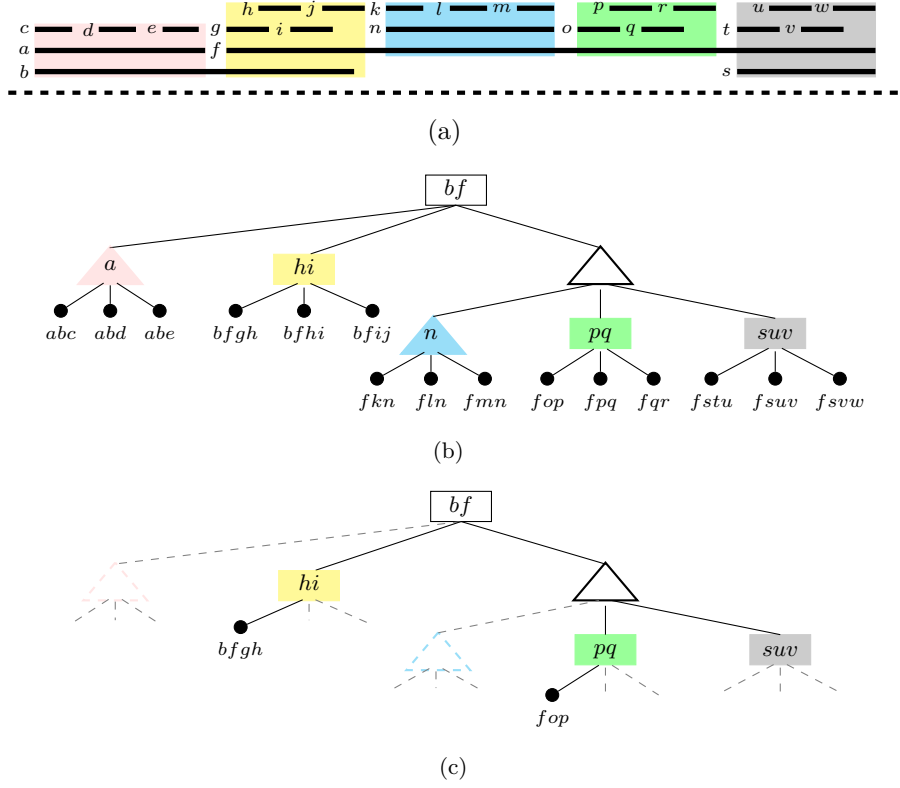
\begin{figure}[tb]
		\centering
		\begin{subfigure}[t]{1\linewidth}
			\centering
			\begin{tikzpicture}[scale=1.4]
				
				\draw [draw=pink, fill=pink, opacity=0.4] (-4,0.15) rectangle (-2.4,0.65);
				
				\draw [draw=yellow, fill=yellow, opacity=0.4] (-2.2,0.15) rectangle (-0.9,0.85);
				
				\draw [draw=cyan, fill=cyan, opacity=0.4] (-0.7,0.85) rectangle (0.88,0.35);
				
				\draw [draw=green, fill=green, opacity=0.4] (1.1,0.35) rectangle (2.4,0.85);
				
				\draw [draw=gray, fill=gray, opacity=0.4] (2.6,0.85) rectangle (3.9,0.15);
				
				\draw[black, dashed, ultra thick] (-4.25,0) -- (4.15,0);
				
				\draw[black, line width=2.0] (-4,0.2) -- (-1,0.2);
				\node (b) at (-4.1,0.2) {\scriptsize{$b$}};
				
				\draw[black, line width=2.0] (-4,0.4) -- (-2.4,0.4);
				\node (a) at (-4.1,0.4) {\scriptsize{$a$}};
				
				\draw[black, line width=2.0] (-4,0.6) -- (-3.65,0.6);
				\node (c) at (-4.1,0.6) {\scriptsize{$c$}};
				
				\draw[black, line width=2.0] (-3.4,0.6) -- (-3.05,0.6);
				\node (d) at (-3.5,0.6) {\scriptsize{$d$}};
				
				\draw[black, line width=2.0] (-2.8,0.6) -- (-2.46,0.6);
				\node (e) at (-2.9,0.6) {\scriptsize{$e$}};
				
				\draw[black, line width=2.0] (-1.9,0.8) -- (-1.5,0.8);
				\node (h) at (-2,0.8) {\scriptsize{$h$}};
				
				\draw[black, line width=2.0] (-1.3,0.8) -- (-0.9,0.8);
				\node (j) at (-1.4,0.8) {\scriptsize{$j$}};
				
				\draw[black, line width=2.0] (-2.2,0.6) -- (-1.8,0.6);
				\node (g) at (-2.3,0.6) {\scriptsize{$g$}};
				
				\draw[black, line width=2.0] (-1.6,0.6) -- (-1.2,0.6);
				\node (i) at (-1.7,0.6) {\scriptsize{$i$}};
				
				\draw[black, line width=2.0] (-2.2,0.4) -- (3.9,0.4);
				\node (f) at (-2.3,0.4) {\scriptsize{$f$}};
				
				\draw[black, line width=2.0] (-0.7,0.8) -- (-0.35,0.8);
				\node (k) at (-0.8,0.8) {\scriptsize{$k$}};
				
				\draw[black, line width=2.0] (-0.1,0.8) -- (0.25,0.8);
				\node (k) at (-0.2,0.8) {\scriptsize{$l$}};
				
				\draw[black, line width=2.0] (0.53,0.8) -- (0.88,0.8);
				\node (m) at (0.39,0.8) {\scriptsize{$m$}};
				
				\draw[black, line width=2.0] (-0.7,0.6) -- (0.88,0.6);
				\node (n) at (-0.8,0.6) {\scriptsize{$n$}};
				
				\draw[black, line width=2.0] (1.1,0.6) -- (1.5,0.6);
				\node (o) at (1,0.6) {\scriptsize{$o$}};
				
				\draw[black, line width=2.0] (1.4,0.8) -- (1.8,0.8);
				\node (p) at (1.3,0.8) {\scriptsize{$p$}};
				
				\draw[black, line width=2.0] (1.7,0.6) -- (2.1,0.6);
				\node (q) at (1.6,0.6) {\scriptsize{$q$}};
				
				\draw[black, line width=2.0] (2,0.8) -- (2.4,0.8);
				\node (r) at (1.9,0.8) {\scriptsize{$r$}};
				
				\draw[black, line width=2.0] (2.6,0.2) -- (3.9,0.2);
				\node (s) at (2.5,0.2) {\scriptsize{$s$}};
				
				\draw[black, line width=2.0] (2.6,0.6) -- (3,0.6);
				\node (t) at (2.5,0.6) {\scriptsize{$t$}};
				
				\draw[black, line width=2.0] (2.9,0.8) -- (3.3,0.8);
				\node (u) at (2.8,0.8) {\scriptsize{$u$}};
				
				\draw[black, line width=2.0] (3.2,0.6) -- (3.6,0.6);
				\node (v) at (3.1,0.6) {\scriptsize{$v$}};
				
				\draw[black, line width=2.0] (3.5,0.8) -- (3.9,0.8);
				\node (w) at (3.4,0.8) {\scriptsize{$w$}};
				
				\node (111) at (-0.15,-0.38) {(a)};	
				
				\label{a}
			\end{tikzpicture}
			
			\label{fig:intervals}
		\end{subfigure}
		~
		\begin{subfigure}[t]{1\linewidth}
			\centering
			\begin{tikzpicture}[yscale=0.5]\footnotesize
				
				\draw [draw=black] (4.2,-0.9) rectangle (5,-1.7);
				\node[label=below:{$\scriptsize{bf}$}] at (4.6,-0.62) [circle, draw, opacity=0, inner sep=0.6mm] (rbf) {};
				
				\node at (-0.05,-2.85) [circle, draw, opacity=0, inner sep=0.6mm] (L) {};
				\draw (4.6,-1.7) -- (L) node[midway, above, 	thick]{};
				
				\node at (2.6,-3.05) [circle, draw, opacity=0, inner sep=0.6mm] (M) {};
				\draw (4.6,-1.7) -- (M) node[midway, above, 	thick]{};
				
				\node at (6.62,-2.9) [circle, draw, opacity=0, inner sep=0.6mm] (R) {};
				\draw (4.6,-1.7) -- (R) node[midway, above, 	thick]{};

				\node[isosceles triangle,
				isosceles triangle apex angle=80,
				draw,
				rotate=90,
				color=pink, 
				fill=pink, 
				thick,
				opacity=0.4,
				minimum size =0.5cm] (T1)at (0.05,-3.45){};
				
				\node[label=below:{$\scriptsize{a}$}] at (0.05,-2.75) [circle, draw, opacity=0, inner sep=0.6mm] (ab) {};
				
				\node at (0.05,-3.65) [draw, fill=black, 	opacity=0, color=black, inner sep=0.8mm] (abx) {};
				
				\node[label=below:{\scriptsize${abc}$}] at (-0.6,-4.5) [circle, draw, fill=black, opacity=1, color=black, inner sep=0.6mm] (abc) {};
				\draw (abx) -- (abc) node[midway, above, 	thick]{}; 
				
				\node[label=below:{\scriptsize${abd}$}] at (0.05,-4.5) [circle, draw, fill=black, opacity=1, color=black, inner sep=0.6mm] (abd) {};
				\draw (abx) -- (abd) node[midway, above, 	thick]{};
				
				\node[label=below:{\scriptsize${abe}$}] at (0.7,-4.5) [circle, draw, fill=black, opacity=1, color=black, inner sep=0.6mm] (abe) {};
				\draw (abx) -- (abe) node[midway, above, 	thick]{};
				
				\draw [draw=yellow, fill=yellow, opacity=0.4] (2.2,-3.8) rectangle (3,-3);
				
				\node[label=below:{$\scriptsize{hi}$}] at (2.6,-2.7) [circle, draw, fill=black, opacity=0, color=black, inner sep=0.6mm] (bf) {};
				
				\node at (2.6,-3.7) [circle, draw, fill=black, opacity=0, color=black, inner sep=0.6mm] (bfx) {};
				
				\node[label=below:{\scriptsize${bfgh}$}] at (1.7,-4.5) [circle, draw, fill=black, opacity=1, color=black, inner sep=0.6mm] (bfgh) {};
				\draw (bfx) -- (bfgh) node[midway, above, 	thick]{};
				
				\node[label=below:{\scriptsize${bfhi}$}] at (2.6,-4.5) [circle, draw, fill=black, opacity=1, color=black, inner sep=0.6mm] (bfhi) {};
				\draw (bfx) -- (bfhi) node[midway, above, 	thick]{};
				
				\node[label=below:{\scriptsize${bfij}$}] at (3.5,-4.5) [circle, draw, fill=black, opacity=1, color=black, inner sep=0.6mm] (bfij) {};
				\draw (bfx) -- (bfij) node[midway, above, 	thick]{};
				
				\node[isosceles triangle,
				isosceles triangle apex angle=80,
				draw,
				rotate=90,
				color=black, 
				thick,
				opacity=1,
				minimum size =0.5cm] (T1)at (6.5,-3.45){};
				
				\node[label=below:{$\scriptsize{}$}] at (6.47,-2.7) [circle, draw, opacity=0, inner sep=0.6mm] (f) {};
				
				\node at (6.5,-3.76) [draw, fill=black, 	opacity=0, color=black, inner sep=0.8mm] (fx) {};
				
				\node[isosceles triangle,
				isosceles triangle apex angle=80,
				draw,
				rotate=90,
				color=cyan, 
				fill=cyan, 
				thick,
				opacity=0.4,
				minimum size =0.5cm] (T2)at (4.3,-5.3){};
				
				\node[label=below:{$\scriptsize{n}$}] at (4.27,-4.6) [circle, draw, opacity=0, inner sep=0.6mm] (fn) {};
				
				\node at (4.27,-5.5) [draw, fill=black, 	opacity=0, color=black, inner sep=0.8mm] (fnx) {};
				\draw (4.3,-4.65) -- (fx) node[midway, above, 	thick]{};
				
				\node[label=below:{\scriptsize${fkn}$}] at (3.56,-6.3) [circle, draw, fill=black, opacity=1, color=black, inner sep=0.6mm] (fk) {};
				\draw (fnx) -- (fk) node[midway, above, 	thick]{};
				
				\node[label=below:{\scriptsize${fln}$}] at (4.27,-6.3) [circle, draw, fill=black, opacity=1, color=black, inner sep=0.6mm] (fl) {};
				\draw (fnx) -- (fl) node[midway, above, 	thick]{}; 
				
				\node[label=below:{\scriptsize${fmn}$}] at (4.98,-6.3) [circle, draw, fill=black, opacity=1, color=black, inner sep=0.6mm] (fm) {};
				\draw (fnx) -- (fm) node[midway, above, 	thick]{};
				
				\draw [draw=green, fill=green, opacity=0.4] (6.1,-5.5) rectangle (6.9,-4.7);
				
				\node[label=below:{$\scriptsize{pq}$}] at (6.5,-4.55) [circle, draw, opacity=0, inner sep=0.6mm] (fpq) {};
				
				\node at (6.5,-5.4) [draw, fill=black, 	opacity=0, color=black, inner sep=0.8mm] (fpqx) {};
				\draw (6.5,-4.7) -- (fx) node[midway, above, 	thick]{};
				
				\node[label=below:{\scriptsize${fop}$}] at (5.8,-6.3) [circle, draw, fill=black, opacity=1, color=black, inner sep=0.6mm] (fno) {};
				\draw (fpqx) -- (fno) node[midway, above, 	thick]{};
				
				\node[label=below:{\scriptsize${fpq}$}] at (6.5,-6.3) [circle, draw, fill=black, opacity=1, color=black, inner sep=0.6mm] (fpq) {};
				\draw (fpqx) -- (fpq) node[midway, above, 	thick]{};
				
				\node[label=below:{\scriptsize${fqr}$}] at (7.2,-6.3) [circle, draw, fill=black, opacity=1, color=black, inner sep=0.6mm] (fqr) {};
				\draw (fpqx) -- (fqr) node[midway, above, 	thick]{};
				
				\draw [draw=gray, fill=gray, opacity=0.4] (8.4,-5.5) rectangle (9.2,-4.7);
				
				\node[label=below:{$\scriptsize{suv}$}] at (8.8,-4.6) [circle, draw, opacity=0, inner sep=0.6mm] (fsuv) {};
				
				\node at (8.8,-5.4) [draw, fill=black, 	opacity=0, color=black, inner sep=0.8mm] (fsuvx) {};
				
				\draw (8.8,-4.7) -- (fx) node[midway, above, 	thick]{};
				
				\node[label=below:{\scriptsize${fstu}$}] at (7.96,-6.3) [circle, draw, fill=black, opacity=1, color=black, inner sep=0.6mm] (fstu) {};
				\draw (fsuvx) -- (fstu) node[midway, above, 	thick]{};
				
				\node[label=below:{\scriptsize${fsuv}$}] at (8.8,-6.3) [circle, draw, fill=black, opacity=1, color=black, inner sep=0.6mm] (fsuv) {};
				\draw (fsuvx) -- (fsuv) node[midway, above, 	thick]{};
				
				\node[label=below:{\scriptsize${fsvw}$}] at (9.65,-6.3) [circle, draw, fill=black, opacity=1, color=black, inner sep=0.6mm] (fsvw) {};
				\draw (fsuvx) -- (fsvw) node[midway, above, 	thick]{};
				
				\node (11) at (4.5,-8.25) {(b)};		\label{b}
			\end{tikzpicture}
			
			\label{fig:PQ-treeT}
		\end{subfigure}
		~
		\begin{subfigure}[t]{1\linewidth}
			\centering
			\begin{tikzpicture}[yscale=0.5]\footnotesize
				
				\draw [draw=black] (4.2,-0.9) rectangle (5,-1.7);
				\node[label=below:{$\scriptsize{bf}$}] at (4.6,-0.62) [circle, draw, opacity=0, inner sep=0.6mm] (rbf) {};
				
				\node at (-0.05,-2.85) [circle, draw, opacity=0, inner sep=0.6mm] (L) {};
				\draw[dashed,gray] (4.6,-1.7) -- (L) node[midway, above, 	thick]{};
				
				\node at (2.6,-3.05) [circle, draw, opacity=0, inner sep=0.6mm] (M) {};
				\draw (4.6,-1.7) -- (M) node[midway, above, 	thick]{};
				
				\node at (6.62,-2.9) [circle, draw, opacity=0, inner sep=0.6mm] (R) {};
				\draw (4.6,-1.7) -- (R) node[midway, above, 	thick]{};
				
				\node[isosceles triangle,
				isosceles triangle apex angle=80,
				draw,dashed,
				rotate=90,
				color=pink, 
				thick,
				opacity=0.4,
				minimum size =0.5cm] (T1)at (0.05,-3.45){};

				\node at (0.05,-3.65) [draw, fill=black, 	opacity=0, color=black, inner sep=0.8mm] (abx) {};
				
				\draw[dashed,gray] (abx) -- (abc) node[midway, above, 	thick]{}; 
				
				\draw[dashed,gray] (abx) -- (abd) node[midway, above, 	thick]{};
				
				\draw[dashed,gray] (abx) -- (abe) node[midway, above, 	thick]{};
				
				\draw [draw=yellow, fill=yellow, opacity=0.4] (2.2,-3.8) rectangle (3,-3);
				
				\node[label=below:{$\scriptsize{hi}$}] at (2.6,-2.7) [circle, draw, fill=black, opacity=0, color=black, inner sep=0.6mm] (bf) {};
				
				\node at (2.6,-3.7) [circle, draw, fill=black, opacity=0, color=black, inner sep=0.6mm] (bfx) {};
				
				\node[label=below:{\scriptsize${bfgh}$}] at (1.7,-4.5) [circle, draw, fill=black, opacity=1, color=black, inner sep=0.6mm] (bfgh) {};
				\draw (bfx) -- (bfgh) node[midway, above, 	thick]{};
				
				\draw[dashed,gray] (bfx) -- (bfhi) node[midway, above, 	thick]{};
				
				\draw[dashed,gray] (bfx) -- (bfij) node[midway, above, 	thick]{};

				\node[isosceles triangle,
				isosceles triangle apex angle=80,
				draw,
				rotate=90,
				color=black, 
				thick,
				opacity=1,
				minimum size =0.5cm] (T1)at (6.5,-3.45){};
				
				\node[label=below:{$\scriptsize{}$}] at (6.47,-2.7) [circle, draw, opacity=0, inner sep=0.6mm] (f) {};
				
				\node at (6.5,-3.76) [draw, fill=black, 	opacity=0, color=black, inner sep=0.8mm] (fx) {};
				
				\node[isosceles triangle,
				isosceles triangle apex angle=80,
				draw,dashed,
				rotate=90,
				color=cyan, 
				thick,
				opacity=0.4,
				minimum size =0.5cm] (T2)at (4.3,-5.3){};

				\node at (4.27,-5.5) [draw, fill=black, 	opacity=0, color=black, inner sep=0.8mm] (fnx) {};
				\draw[dashed,gray] (4.3,-4.65) -- (fx) node[midway, above, 	thick]{};
				
				\draw[dashed,gray] (fnx) -- (fk) node[midway, above, 	thick]{};
				
				\draw[dashed,gray] (fnx) -- (fl) node[midway, above, 	thick]{}; 
				
				\draw[dashed,gray] (fnx) -- (fm) node[midway, above, 	thick]{};
				
				\draw [draw=green, fill=green, opacity=0.4] (6.1,-5.5) rectangle (6.9,-4.7);
				
				\node[label=below:{$\scriptsize{pq}$}] at (6.5,-4.55) [circle, draw, opacity=0, inner sep=0.6mm] (fpq) {};
				
				\node at (6.5,-5.4) [draw, fill=black, 	opacity=0, color=black, inner sep=0.8mm] (fpqx) {};
				\draw (6.5,-4.7) -- (fx) node[midway, above, 	thick]{};
				
				\node[label=below:{\scriptsize${fop}$}] at (5.8,-6.3) [circle, draw, fill=black, opacity=1, color=black, inner sep=0.6mm] (fno) {};
				\draw (fpqx) -- (fno) node[midway, above, 	thick]{};
				
				\draw[dashed,gray] (fpqx) -- (6.5,-6.3) node[midway, above, 	thick]{};
				
				\draw[dashed,gray] (fpqx) -- (7.2,-6.3) node[midway, above, 	thick]{};
				
				\draw [draw=gray, fill=gray, opacity=0.4] (8.4,-5.5) rectangle (9.2,-4.7);
				
				\node[label=below:{$\scriptsize{suv}$}] at (8.8,-4.6) [circle, draw, opacity=0, inner sep=0.6mm] (fsuv) {};
				
				\node at (8.8,-5.4) [draw, fill=black, 	opacity=0, color=black, inner sep=0.8mm] (fsuvx) {};
				
				\draw (8.8,-4.7) -- (fx) node[midway, above, 	thick]{};
				
				\draw[dashed,gray] (fsuvx) -- (fstu) node[midway, above, 	thick]{};
				
				\draw[dashed,gray] (fsuvx) -- (8.8,-6.3) node[midway, above, 	thick]{};
				
				\draw[dashed,gray] (fsuvx) -- (fsvw) node[midway, above, 	thick]{};
				
				\node (11) at (4.5,-8.25) {(c)};		\label{c}
			\end{tikzpicture}
			
			\label{fig:PQ-treeT}
		\end{subfigure}
		
		\caption{(a) Top: an interval graph $G$ on the vertex set $V(G)=\{a,b,\ldots,w\}$, where the vertices are labelled on the left of the intervals in its representation.
		(b) A PQ-tree $T$ of $G$ with the internal nodes (P-nodes are triangle-shaped and Q-nodes are rectangle-shaped) labelled with lists of the corresponding inner vertices (cf.~\Cref{sec:intPQ}), and with the leaves labelled by the maximal cliques~of~$G$. 
		(c) For an illustration, let $\ca A=\{\{b,f,g\},\{f,o,p\},\{s,u,v\}\}$ be a family of three marked sets in~$G$.
		Then the bottom picture shows the subtree $T'$ obtained from $T$ by discarding all clean subtrees (wrt.~$\ca A$, now only outlined in dashed lines); this is used in the proof of Lemma~\ref{lem:markedtosets}.}
		\label{fig:lemma13}
	\end{figure}

	Our answer to handling PQ-trees with marked sets is to precompute the complete isomorphism types of all clean subtrees in $T$,
	and then to introduce new marked sets (in addition to~$\ca A$) which encode the (few) non-clean subtrees of $T$.
	Importantly, the latter can be done without increasing the maximum antichain size of the marked sets.
	Altogether, this is formulated in detail as follows:
	
	\begin{lemma}\label{lem:markedtosets}
		Assume an instance of the problem {\sc AutomMarkedINT}$(G;\,\ca A^1,\ldots,\ca A^m)$, where the maximum antichain size of $\ca A:=\ca A^1\cup\ldots\cup\ca A^m$ equals~$a$.
		Then there are families $\ca B^1\cup\ldots\cup\ca B^{k}=\ca A$ of subsets of $V(G)$ where $(\ca B^1,\ldots,\ca B^{k})$ is a partition of $\ca A$ refining $(\ca A^1,\ldots,\ca A^{m})$,
		and another family $\ca C\subseteq2^{V(G)}$ of nonempty subsets of vertices where $\ca C$ is partitioned into families $(\ca C^1,\ldots,\ca C^\ell)$, such that the following holds.
		If a group $\Gamma$ is the solution to the problem {\sc AutomSET}$(V(G);$ $\ca B^1,\ldots,\ca B^k,\,\ca C^1,\ldots,\ca C^\ell)$, 
		then $\Gamma$ restricted to $\ca A$ is the sought solution of {\sc AutomMarkedINT}$(G;$ $\ca A^1,\ldots,\ca A^m)$.
		Furthermore, the maximum antichain size in $\ca A\cup\ca C^1\cup\ldots\cup\ca C^\ell$ is also $a$,
		and the families $\ca B^1,\ldots,\ca B^k$ and $\ca C^1,\ldots,\ca C^\ell$ {can be computed in linear time.}
	\end{lemma}
	
	\begin{proof}
		Let $T$ be a PQ-tree of our interval graph $G$. 
		We first show how to refine the partition $(\ca A^1,\ldots,\ca A^{m})$ of $\ca A$, using an auxiliary annotation assigned to the sets of~$\ca A$.
		Observe that since every set $A\in\ca A$ induces a clique in $G$, the set $A$ cannot intersect the inner vertices of two nodes of $T$ which are incomparable in the tree order.
		This holds since $A$ must be a subset of some maximal clique of $G$, and so if $A$ intersects some inner vertex assigned to a node $q$, 
		the max clique containing $A$ must be among the descendant leaves of $q$ in~$T$ (it may be useful to imagine this fact in an actual interval representation of~$G$).
		
		Consequently, the nodes with assigned inner vertices from $A\in\ca A$ lie on a root-to-leaf path of $T$, and so it is sound (and automorphism-invariant) 
		to annotate $A$ with the numbers of its elements which are the inner vertices assigned to nodes at each level of $T$ from the root.
		For Q-nodes of $T$, we additionally annotate $A$ with the numbers of elements which are inner vertices of each rank at every level.
		Clearly, by uniqueness of a PQ-tree up to transformations, if an automorphism of $G$ maps $A\in\ca A$ into $A'\in\ca A$, then the annotations of $A$ and of $A'$ are the same.
		
		\smallskip
		Then we ``reduce'' the PQ-tree $T$ of $G$ to a tree $T'\subseteq T$ by discarding all clean subtrees (with all their nodes) from it. See~\Cref{fig:lemma13}(c).
		Moreover, for every node $q$ of $T'$, we denote by $T_q$ the subtree of $T$ rooted at $q$ and formed by all clean subtrees of~$q$.
		We compute by Corollary~\ref{cor:canonint} the canonical PQ-tree $T^o_q$ of the subgraph $G_q$ represented by~$T_q$, and store $T^o_q$ with the implied presentation of the (unlabelled) graph $G_q$ as an annotation of the node $q$ in~$T'$.
		If $q$ is a Q-node, we additionally store in the annotation of $q$ the appropriate positions of the non-clean subtrees of $q$ in $T$ within the order of the sons of $q$ in $T^o_q$.
		
		For $q\in V(T')$, we now define the set $C_q$ as the vertices of $G$ which belong to $q$ in $T$, i.e., $C_q$ equals the union of all maximal cliques (of~$G$) being in the descendant leaves of~$q$.
		Since $q$ belongs also to $T'$, it is not clean, and so there is a set $A\in\ca A$ intersecting the inner vertices assigned to~$q$.
		Observe that, whenever $A\in\ca A$ contains some of the inner vertices assigned to~$q$, then~$A\subseteq C_q$.
		This holds, as above, since $A$ is a subset of some maximal clique of $G$, and this clique containing $A$ must be among the descendant leaves of $q$ in~$T$.
		
		Let $\ca C:=\{C_q: q\in V(T')\}$. We claim that the maximum antichain size in $\ca A\cup\ca C$ is~$a$ (as in $\ca A$ itself).
		To prove this claim, let $\ca D\subseteq\ca A\cup\ca C$ be an antichain such that $|\ca D\cap\ca C|$ is the least possible among antichains of the same size.
		If $\ca D\subseteq\ca A$, we are done, and so let $C_q\in\ca D\cap\ca C$ for some node $q\in V(T')$.
		By the previous, there exists $A_q\in\ca A$ such that $A_q$ intersects some of the inner vertices of $G$ assigned to $q$ and $A_q\subseteq C_q$.
		Since $\ca D\ni C_q$ is an antichain, no set from $\ca D$ is contained in~$A_q$. 
		Assume that there is some $A\in\ca D\cap\ca A$ such that $A\supsetneq A_q$. Then $A$ contains an inner vertex of $q$, and hence $A\subseteq C_q$ which contradicts $\ca D$ being an antichain.
		Finally, assume that there is $C_{q'}\in\ca D\cap\ca C$, where $q'\in V(T')$, such that~$C_{q'}\supseteq A_q$.
		Then $q'$ is not a descendant of $q$ since $C_{q'}\not\subseteq C_q$ and, likewise, $q'$ is not an ancestor of~$q$.
		However, $C_{q'}$ contains an inner vertex of $A$ assigned to $q$ which contradicts the definition of inner vertices.
		Therefore, $(\ca D\setminus\{C_q\})\cup\{A_q\}$ is an antichain again. Overall, we get that indeed $\ca D\subseteq\ca A$, and so~$|\ca D|\leq a$.
		
		\smallskip
		We refine the partition $(\ca A^1,\ldots,\ca A^m)$ of $\ca A$ into wanted $(\ca B^1,\ldots,\ca B^{k})$ (for appropriate $k\geq m$) according to the above annotation;
		for $i=1,\ldots,m$, two sets $A_1,A_2\in\ca A^i$ fall into the same part -- color, of $(\ca B^1,\ldots,\ca B^{k})$ if and only if the annotations of $A_1$ and of $A_2$ in $T$ are the same.
		Likewise, we partition $\ca C$ into subfamilies $(\ca C^1,\ldots,\ca C^\ell)$ (with an arbitrary number $\ell$ of parts) as follows:
		two sets $C_q$ and $C_{q'}$ fall into the same part -- color, if and only if the nodes $q$ and $q'$ have received in $T'$ above the same annotation (in particular, the same canonical PQ-tree of the clean subtrees of~$q$~or~$q'$).
		
		We now verify the claimed properties of the defined partitions of vertex set families.
		Recall that the permutation group $\Gamma$ is the solution to {\sc AutomSET}$(V(G);\,\ca B^1,\ldots,\ca B^k,\,\ca C^1,\ldots,\ca C^\ell)$.
		Let a permutation group $\Delta$ be the solution of {\sc AutomMarkedINT}$(G;\,\ca A^1,\ldots,\ca A^m)$.
		
		For every permutation $\alpha\in\Delta$ of $\ca A$, there exists (by Definition~\ref{def:autommarked}) an underlying $(\ca A^1,\ldots,\ca A^m)$-preserving automorphism $\beta$ of~$G$.
		Then $\beta$ induces a permutation on the maximal cliques of $G$ which, in turn, gives an automorphism $\beta'$ of the PQ-tree~$T$.
		Note that in this context, naturally, an automorphism of a PQ-tree not only preserves the underlying rooted tree, but also the order of every Q-node up to reversal.
		As noted above, since $\beta$ is an automorphism, $\alpha$ must preserve the annotations of the sets of $\ca A$, i.e., the partition $(\ca B^1,\ldots,\ca B^{k})$.
		Since clean subtrees are also preserved by $\beta'$, the restriction of $\beta'$ gives an annotation-preserving automorphism of the tree~$T'$ by the definition of a canonical PQ-tree.
		Hence $\beta'$ induces a permutation of $\ca C$ respecting the partition $(\ca C^1,\ldots,\ca C^\ell)$, which composed together with $\alpha$ on $\ca A$ gives a unique permutation~$\gamma\in\Gamma$.
		
		Conversely, let $\gamma\in\Gamma$ be a permutation on~$\ca A\cup\ca C$ respecting both partitions $(\ca B^1,\ldots,\ca B^{k})$ and $(\ca C^1,\ldots,\ca C^\ell)$.
		In particular, every set $C_q\in\ca C^i\subseteq\ca C$, where $q\in V(T')$ and $1\leq i\leq\ell$, is mapped into $C_{r}=\gamma(C_q)\in\ca C^i$ where~$r\in V(T')$.
		Henceforth, $\gamma$ induces (as~`$q\mapsto r$') an annotation-preserving permutation $\beta_0$ of $V(T')$.
		Since we have, for any two nodes $p,p'\in V(T')$, that $p$ is an ancestor of $p'$ if and only if $C_p\supseteq C_{p'}$,
		and $\gamma$ preserves the inclusion relation, we conclude that $\beta_0$ is an annotation-preserving automorphism of~$T'$.
		Since our annotation at every node of $T'$ includes the canonical PQ-tree of the clean subtrees, and the ordering of the non-clean subtrees under Q-nodes, 
		$\beta_0$ extends to a permutation $\beta_1$ of whole $V(T)$ which hence is an automorphism of the whole PQ-tree~$T$.
		This gives an underlying automorphism $\beta$ of the graph~$G$. 
		Since, moreover, $\gamma$ preserves the annotations of the sets in~$\ca A$, the automorphism $\beta$ can be chosen such that it agrees with the permutation $\gamma$ on $\ca A$.
		We have got the restriction of $\gamma$ to $\ca A$ in the group~$\Delta$.
		
		\smallskip
		It remains to analyze the runtime of the described reduction, i.e., of the computation of the families $\ca B^1,\ldots,\ca B^k$ and $\ca C^1,\ldots,\ca C^\ell$.
		For that we first compute a PQ-tree $T$ of our interval graph $G$ in linear time by~\Cref{thm:recogIntervalLinear}.
		From $T$, we straightforwardly compute the anntations of the sets in $\ca A$, as specified above.
		Then we easily in linear time identify all clean subtrees of $T$, and hence get the tree~$T'$.
		The annotations of the nodes of $T'$ are computed again in linear time using Corollary~\ref{cor:canonint} applied to their clean subtrees,
		and since the considered clean subtrees in this computation are pairwise disjoint, the overal runtime is linear in the size of~$G$.
		Knowing the annotations, we then easily output the families $\ca B^1,\ldots,\ca B^k$ and $\ca C^1,\ldots,\ca C^\ell$.
	\end{proof}

	\begin{proof}[Proof of \Cref{thm:autommarked}]
		We first apply Lemma~\ref{lem:markedtosets}; in this way, we transform an instance of {\sc AutomMarkedINT}$(G;\,\ca A^1,\ldots,\ca A^m)$
		into an instance of {\sc AutomSET}$(V(G);$ $\ca B^1,\ldots,\ca B^k,$ $\ca C^1,\ldots,\ca C^\ell)$ of the same parameter value~$a$.
		Then, we use \Cref{thm:automsetfam} to solve the latter in FPT time with respect to~$a$, and straightforwardly output the
		corresponding restricted solution of {\sc AutomMarkedINT}$(G;\,\ca A^1,\ldots,\ca A^m)$.
	\end{proof}

	\section{Conclusions}

	We have introduced the problem {\sc AutomMarkedINT}$(G;\,\ca A^1,\ldots,\ca A^m)$ which is a clean and rigorous
	new general formulation of an algorithmic task previously used in isomorphism algorithms for special classes of chordal graphs 
	\cite{aaolu2019isomorphism,DBLP:conf/walcom/CagiriciH22}.
	We believe that this self-contained exposition of the solution of {\sc AutomMarkedINT}$(G;\,\ca A^1,\ldots,\ca A^m)$
	can be interesting and useful on its own, not only as a minor technical tool.
	For instance, we think it can be useful in development of an isomorphism algorithm for so-called $H$-graphs
	(these are intersection graphs of connected subgraphs in a suitable subdivision of the base graph $H$, and they naturally generalize interval graphs and chordal graphs of bounded leafage)
	in the case that $H$ contains one cycle (while for $H$ containing more than one cycle the problem is known to be GI-complete).

\bibliography{Union-bibliography}

\begin{thebibliography}{1}

\bibitem{AHU}
Alfred~V. Aho, John~E. Hopcroft, and Jeffrey~D. Ullman.
\newblock {\em The Design and Analysis of Computer Algorithms}.
\newblock Addison-Wesley, 1974.

\bibitem{aaolu2019isomorphism}
Deniz A\u{g}ao\u{g}lu and Petr Hlin\v{e}n{\'{y}}.
\newblock Isomorphism problem for {$S_d$}-graphs.
\newblock In Javier Esparza and Daniel Kr{\'{a}}l', editors, {\em 45th
  International Symposium on Mathematical Foundations of Computer Science,
  {MFCS} 2020, August 24-28, 2020, Prague, Czech Republic}, volume 170 of {\em
  LIPIcs}, pages 4:1--4:14. Schloss Dagstuhl - Leibniz-Zentrum f{\"{u}}r
  Informatik, 2020.
\newblock \href {https://doi.org/10.4230/LIPIcs.MFCS.2020.4}
  {\path{doi:10.4230/LIPIcs.MFCS.2020.4}}.

\bibitem{DBLP:conf/walcom/CagiriciH22}
Deniz {A\u{g}ao\u{g}lu {\c{C}}agirici} and Petr Hlin\v{e}n{\'{y}}.
\newblock Isomorphism testing for {T}-graphs in {FPT}.
\newblock In {\em {WALCOM}}, volume 13174 of {\em Lecture Notes in Computer
  Science}, pages 239--250. Springer, 2022.

\bibitem{babai-bdcm}
L{\'{a}}szl{\'{o}} Babai.
\newblock Monte {C}arlo algorithms in graph isomorphism testing.
\newblock {\em Tech.~Rep.~79-10, Universit\'e de Montr\'eal}, 1979.
\newblock 42 pages.

\bibitem{DBLP:conf/stoc/Babai16}
L{\'{a}}szl{\'{o}} Babai.
\newblock Graph isomorphism in quasipolynomial time [extended abstract].
\newblock In Daniel Wichs and Yishay Mansour, editors, {\em Proceedings of the
  48th Annual {ACM} {SIGACT} Symposium on Theory of Computing, {STOC} 2016,
  Cambridge, MA, USA, June 18-21, 2016}, pages 684--697. {ACM}, 2016.
\newblock \href {https://doi.org/10.1145/2897518.2897542}
  {\path{doi:10.1145/2897518.2897542}}.

\bibitem{recogIntervalLinear}
Kellogg~S. Booth and George~S. Lueker.
\newblock Testing for the consecutive ones property, interval graphs, and graph
  planarity using {PQ}-tree algorithms.
\newblock {\em J. Comput. Syst. Sci.}, 13(3):335--379, 1976.
\newblock \href {https://doi.org/10.1016/S0022-0000(76)80045-1}
  {\path{doi:10.1016/S0022-0000(76)80045-1}}.

\bibitem{AutMPQtrees}
Charles~J. Colbourn and Kellogg~S. Booth.
\newblock Linear time automorphism algorithms for trees, interval graphs, and
  planar graphs.
\newblock {\em {SIAM} J. Comput.}, 10(1):203--225, 1981.
\newblock \href {https://doi.org/10.1137/0210015} {\path{doi:10.1137/0210015}}.

\bibitem{furst}
Merrick~L. Furst, John~E. Hopcroft, and Eugene~M. Luks.
\newblock Polynomial-time algorithms for permutation groups.
\newblock In {\em 21st Annual Symposium on Foundations of Computer Science,
  Syracuse, New York, USA, 13-15 October 1980}, pages 36--41. {IEEE} Computer
  Society, 1980.
\newblock \href {https://doi.org/10.1109/SFCS.1980.34}
  {\path{doi:10.1109/SFCS.1980.34}}.

\bibitem{planarLinear}
John~E. Hopcroft and J.~K. Wong.
\newblock Linear time algorithm for isomorphism of planar graphs (preliminary
  report).
\newblock In Robert~L. Constable, Robert~W. Ritchie, Jack~W. Carlyle, and
  Michael~A. Harrison, editors, {\em Proceedings of the 6th Annual {ACM}
  Symposium on Theory of Computing, April 30 - May 2, 1974, Seattle,
  Washington, {USA}}, pages 172--184. {ACM}, 1974.
\newblock \href {https://doi.org/10.1145/800119.803896}
  {\path{doi:10.1145/800119.803896}}.

\end{thebibliography}
	
\end{document}